\PassOptionsToPackage{x11names,table,dvipsnames}{xcolor}
\documentclass[acmtog,balance=false]{acmart}

\usepackage{booktabs} % For formal tables
\usepackage{scalerel}
\usepackage{subcaption}
\usepackage{dsfont}
\usepackage{bm}
\usepackage{wrapfig}
\usepackage[ruled]{algorithm2e} % For algorithms
\usepackage{color}
\usepackage{tabularx}

\SetAlFnt{\small}
\SetAlCapFnt{\small}
\SetAlCapNameFnt{\small}
\SetAlCapHSkip{0pt}

\setcopyright{acmcopyright}
\acmJournal{TOG}
\acmYear{2022}
\acmVolume{41}
\acmNumber{4}
\acmArticle{68}
\acmMonth{7}
\acmDOI{10.1145/3528223.3530093}
\definecolor{redcolor}{rgb}{0.8,0,0}
\definecolor{bluecolor}{rgb}{0,0,0.8}
\definecolor{greencolor}{rgb}{0.0,0.5,0.0}

\DeclareMathOperator*{\argmin}{argmin}

\allowdisplaybreaks

\begin{document}
% Title portion
\title{Fast Evaluation of Smooth Distance Constraints on Co-Dimensional Geometry}

\begin{teaserfigure}
  \includegraphics[width=\textwidth]{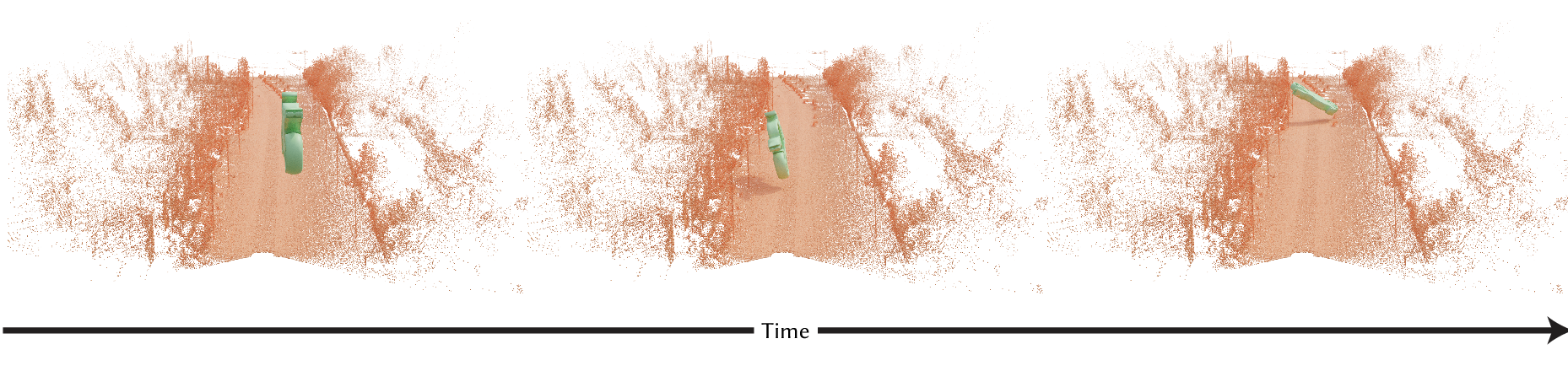}
  \caption{Using a smooth distance function as an intersection-free constraint, we can simulate rigid body contact between a variety of meshes. Here, a motorcycle jumps onto a point cloud street, where it slides and crashes through the street, hitting electrical poles along the way. Collisions are resolved using a single inequality constraint in a primal-dual interior-point solver.}
  \label{fig:bike_city}
\end{teaserfigure}

% DO NOT ENTER AUTHOR INFORMATION FOR ANONYMOUS TECHNICAL PAPER SUBMISSIONS TO SIGGRAPH 2019!
\author{Abhishek Madan}
\affiliation{%
  \institution{University of Toronto}
  \city{Toronto}
  \country{Canada}
}
\email{amadan@cs.toronto.edu}
\author{David I.W. Levin}
\affiliation{%
  \institution{University of Toronto}
  \city{Toronto}
  \country{Canada}
}
\email{diwlevin@cs.toronto.edu}

\newcommand{\bx}{\mathbf{x}}
\newcommand{\by}{\mathbf{y}}
\newcommand{\bp}{\mathbf{p}}
\newcommand{\bv}{\mathbf{v}}
\newcommand{\bq}{\mathbf{q}}
\newcommand{\Md}{\mathcal{M}}
\newcommand{\Mq}{\bar{\mathcal{M}}}
\newcommand{\bV}{\bar{V}}
\newcommand{\bF}{\bar{F}}
\newcommand{\bphi}{\bm{\phi}}
\newcommand{\blambda}{\bm{\lambda}}
\newcommand{\bpi}{\bm{\pi}}
\newcommand{\btheta}{\bm{\theta}}

\begin{abstract}
We present a new method for computing a smooth minimum distance function based on the LogSumExp function for point clouds, edge meshes, triangle meshes, and combinations of all three.
We derive blending weights and a modified Barnes-Hut acceleration approach that ensure our method approximates the true distance, and is conservative (points outside the zero isosurface 
are guaranteed to be outside the surface) and efficient to evaluate for all the above data types.
This, in combination with its ability to smooth sparsely sampled and noisy data, like point clouds, shortens the gap between data acquisition and simulation,
and thereby enables new applications such as direct, co-dimensional rigid body simulation using unprocessed lidar data.
\end{abstract}

%
% TODO:





\maketitle

\section{Introduction}

Distance fields are integral to many applications in computer graphics and scientific computing. 
In rendering, distance fields provide an implicit shape representation that enables both flexible editing and fast display, while in physics simulation they provide a convenient way to represent distance-mediated interactions between simulated objects — such as collisions. 
Any geometric representation can be converted into a distance field, whether it be a point cloud, edge mesh, or triangle mesh. 
Thus, algorithms that rely on distance field representations are theoretically invariant to input geometry type. 
This is important because many applications of geometry processing and physics simulation act on mixed geometric input (e.g., self-driving car simulations represent the cars as polygonal models but the environment is acquired as a point cloud via lidar scan). 

Unfortunately, current distance field representations fall short of living up to these theoretical advantages. 
Storing distance fields on grids is memory intensive and can require costly preprocessing, while fitting neural networks alleviates the memory pressure but requires a much higher upfront cost in training time and data consumption, while also being difficult to generalize.  
Complicating proceedings is the fact that, often, representing the exact distance field is not ideal for practical applications since input geometric models are usually an approximation of the underlying true object. 
Representing any curved, smooth surface using piecewise linear triangles is an obvious example, but noisy or incomplete data, like lidar point clouds, is another. 
While methods exist for accurately simulating the latter (Fig.~\ref{fig:ipc-points} reproduced from \citet{Ferguson:2021:RigidIPC}), the result is not generally applicable  to cases where discrete samples are meant to coalesce into a smooth surface.

\begin{figure}
  \includegraphics[width=\columnwidth]{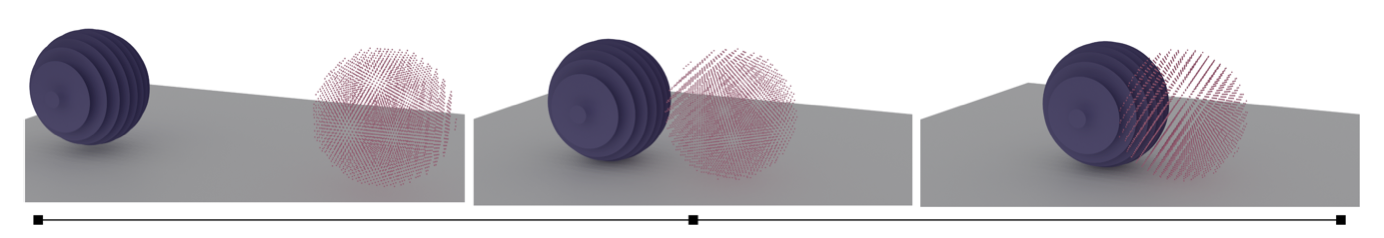}
  \caption{Using Rigid IPC~\cite{Ferguson:2021:RigidIPC}, a collision between a sphere of points and a sphere of disjoint planes causes them to lock together on impact. While an impressive demonstration of robustness, this result is counter intuitive if the point cloud were meant to represent a solid sphere.}
  \label{fig:ipc-points}
\end{figure}

Further complicating matters is the fact that exact distance fields are typically not smooth, which limits the choice of algorithms that can be applied. 
Finally, existing geometry processing pipelines are often set up to receive closed triangle mesh input only~\citep{Li2020IPC}, with co-dimensional inputs (inputs that feature a mixture of triangles, edges, and points) considered special cases~\citep{Li2021CIPC}.

In this paper we present a smooth distance formulation that addresses these issues while simultaneously guaranteeing theoretical properties that are crucial for distance fields to function robustly in rendering and simulation applications. 
Specifically, our method  

\begin{itemize}
  \item is an implicit function, which naturally fits in an optimization context;
  \item is an approximation of the exact unsigned minimum distance between two geometric quantities;
  \item can represent different types of geometry (in this paper we focus on points, edges, and triangles), and can represent both closed and open or co-dimensional geometry;
  \item is smooth (i.e., able to take derivatives), which is beneficial for optimization;
  \item is efficient to evaluate;
  \item conservatively estimates (i.e., underestimates) exact distance.
\end{itemize}

We achieve these goals by using a smooth minimum distance function based on LogSumExp (sometimes called Kreisselmeier-Steinhauser distance),
augmented by weight functions to remove bulge artifacts from edge and triangle mesh isosurfaces, and a conservative Barnes-Hut approximation to speed up function evaluations while inducing slight discontinuities at the near field-far field boundary (see Section~\ref{sec:params}).
We demonstrate the efficacy and ease-of-use of our smooth distance field representation on a number of colliding rigid body simulations which directly act on co-dimensional geometry, including difficult cases such as collision-mediated interaction with lidar data featuring millions of points.
Our approach could be a drop in improvement to many existing computer graphics applications as well as a major step forward for cutting-edge pursuits such as the direct simulation of self-driving cars in lidar environments. 

\section{Related Work}

%\textit{LogSumExp and other smooth distance functions.}
The LogSumExp function is commonly used in deep learning (see, e.g.,~\cite{Zhang2020}) as a smooth estimate of the maximum of a set of data.
Its gradient is the softmax function (which is not a maximum as the name implies, but a smooth estimate of the one-hot argmax of a set of data).
LogSumExps can be easily modified to return a smooth minimum distance rather than a maximum (and its gradient is the softmin).
Aside from deep learning, LogSumExps also appear in other contexts where smooth approximations to min/max functions are needed: for example, they are known in the engineering literature as the Kreisselmeier-Steinhauser distance~\cite{Kreisselmeier1979}.
LogSumExps have also been used recently in computer graphics by \citet{Panetta2017} to smoothly blend between microstructure joints.
LogSumExp is just one of a number of smooth distance functions.
For instance, the $L_p$ norm function has been used as a smooth distance as well, for smooth blending between implicit surfaces~\cite{Wyvill1999} and computing smooth shells around surfaces~\cite{Peng2004}.
A similar function can be computed directly from boundary integrals~\cite{Belyaev2013}.
These functions could act as a drop in replacement for much of our proposed algorithm, but the former function tends to return numerically 0 results for far-away points which makes the zero isosurface ambiguous, and it is unclear if the latter even exhibits the important underestimate property.
Not only do LogSumExps satisfy the desired property, but they numerically return \texttt{inf} for far-away points which leaves the zero isosurface unambiguous, and so we choose to construct our method around the LogSumExp function.

\citet{Gurumoorthy2009} demonstrated a relationship between the Eikonal equation and the time-independent Schr\"odinger equation (which is a screened Poisson equation), and used this to derive the LogSumExp function as an approximate solution to the Eikonal equation. %which converges as the user parameter $h$ approaches $\infty$ (our parameter $\alpha$ is the same as $1/h$).
They evaluate the LogSumExp using a convolution in the frequency domain, which requires a background grid to compute the FFT and its inverse. Further, their method requires all data points to be snapped to grid vertices.
While more efficient than a full evaluation, our method achieves comparable asymptotic performance without a background grid and therefore respects the input geometry.
\citet{Sethi2012} extended this line of work by adding support for edges, but they integrate the exponentiated distance over each edge, which, as we show in Section~\ref{sec:dist}, can lead to overestimated distances.
Computing a distance approximation by taking a logarithm is also conceptually similar to Varadhan's formula geodesic distance, which was the inspiration for the geodesics in heat method~\cite{Crane:2017:HMD}.

% Cited in methods section but placed here for layout
\begin{figure*}
  \includegraphics[width=\textwidth]{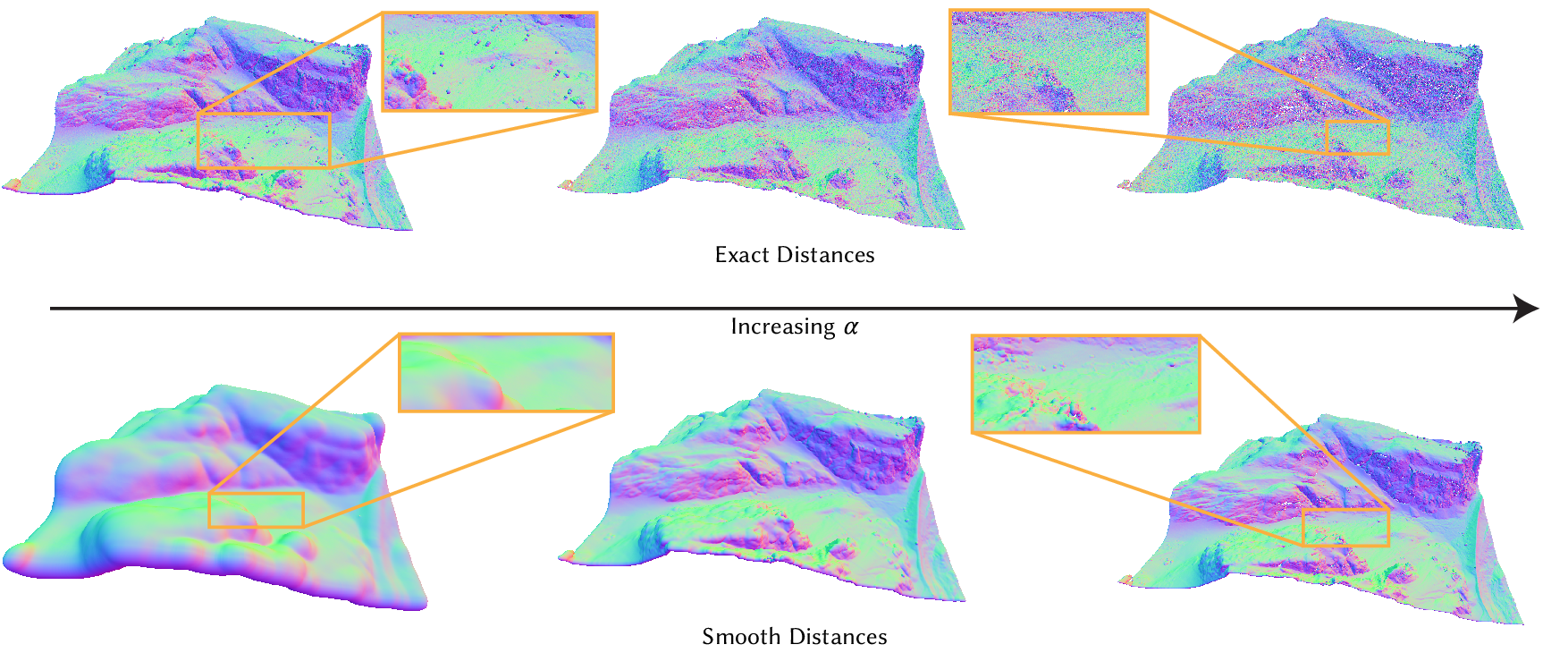}
  \caption{A comparison between smooth and exact (offset) distance isosurfaces, with gradients represented by colour, at varying values of $\alpha$, using $1/\alpha$ for exact distance offets. Exact distances trade off between inflating noisy data in the center and far side of the point cloud at low $\alpha$, and noisy gradients at high $\alpha$. Meanwhile, smooth distances (bottom) transition from smooth and inflated to more exact as $\alpha$ increases, without exacerbating noisy data or yielding poor gradients.}
  \label{fig:exact_smooth}
\end{figure*}

%\am{Should we cite more papers here? Feels like it derails the discussion if we talk too much about neural sdfs so I just cited the ``seminal'' works.}
Smooth signed distance functions (SDFs) have recently become popular in machine learning as well. A full accounting is beyond the scope of this paper but 
see for instance \citet{Chen2019,Park2019,Mescheder2019}.
These approaches encode geometric information in latent vectors to be used at evaluation time, along with an input position in space to evaluate a learned signed distance function.
Aside from being unsigned rather than signed, our distance approximation diverges in two important ways from work on Neural SDFs.
First, our representation is an augmentation to the exact geometry as a smooth approximation rather than an outright replacement. 
This means that algorithms that still require the original discrete geometry~\cite{Li2020} for operations such as continuous collision detection can make use of our method.
Second, the only preprocessing in our method is building a BVH over the data (which typically takes less than a second) rather than training a neural network, which is significantly more expensive.

%\textit{Surface reconstruction from point clouds.}
One particular feature of smooth distance functions is that they can use point data to construct implicit functions whose zero isosurface represents the surface.
There are many other methods that accomplish this, such as radial basis functions~\cite{Carr2001} and Poisson surface reconstruction~\cite{Kazhdan2006,Kazhdan2013}.
Although these methods are capable of producing very accurate surface reconstructions, they require solving large linear systems, while in our approach the implicit function is readily available. %(though it does not have a solid isosurface for large $\alpha$) -- people don't know what alpha is yet.
Another surface reconstruction method is the point set surface~\cite{Alexa2003}, though instead of obtaining an implicit function, this method finds points on the described surface.
Level set methods~\cite{Zhao2001} have also been used to reconstruct surfaces, though this requires a grid and may require prohibitively dense voxels to capture fine detail in the underlying surface.

%\textit{Collision resolution.}
Collision resolution has long been a difficult problem in physics-based animation and engineering.
While no efficient method for deformable SDFs exists, a useful approximation is to use a signed distance function of the undeformed space~\cite{McAdams2011}.
\citet{Mitchell2015} extended this work by using a non-manifold grid to accurately represent high-frequency and even zero-width features.
Recently, \citet{Macklin2020} have used SDFs to represent rigid objects that robustly collide with deformable objects in an extended position-based dynamics framework~\cite{Macklin2016}.
Further, barrier energies for constrained optimization can be seen as a smooth analog for distance-based constraints.
These have seen much success in geometry processing~\cite{10.1145/2766947} and physics simulation~\cite{Li2020,Ferguson:2021:RigidIPC},
and follow-up work has proposed separate extensions to co-dimensional~\cite{Li2021CIPC} and medial geometry~\cite{MedialICP}. 
While effective for this particular application, the methods of \citet{Macklin2020} and \citet{Li2020} lack the ability to smooth input data~\cite{Ferguson:2021:RigidIPC} and so cannot, for instance, approximate point clouds as closed surfaces~(Fig.~\ref{fig:bike_city}) for smooth collision resolution. 
Another shared technical limitation of these methods is the large number of constraints they generate (one per primitive pair), which must be mitigated through techniques like spatial hashing.
For these reasons, we view our method as complementary to the aforementioned approaches: our method smooths the input data while also combining every pairwise primitive constraint into a single constraint.

Simulation frameworks such as Bullet~\cite{Coumans2021} and PhysX~\cite{nvidia2021} often accelerate collisions through bounding proxies that cover sections of geometry such as convex hulls, spheres, and cylinders.
These approaches are effective for real-time simulation where speed is preferred over accuracy, but not comparable with our method since we aim for accurate off-line simulation.

The key to our method is a carefully designed weighted smooth distance function, combined with a specialized Barnes-Hut approximation~\cite{Barnes1986}.
The Barnes-Hut algorithm was first developed for N-body simulations to reduce computational effort on far-away bodies with negligible contributions to force.
Barnes-Hut approximations have seen use in many graphics applications which use rapidly decaying kernels (e.g.,~\cite{Barill2018,Alexa2003,Yu:2021:RC}),
but as we will show, careful modification is needed to ensure that fast evaluation does not break the conservative bounds of the LogSumExp function.

\paragraph{Contributions}
In this paper we present a new method for computing smooth distance fields on co-dimensional input geometry.
We derive blending weights and a modified Barnes-Hut acceleration approach that ensures our method is conservative (points outside the zero isosurface 
are guaranteed to be outside the surface), accurate, and efficient to evaluate for all the above data types.
This, in combination with its ability to smooth sparsely sampled data like point clouds, enables new applications such as direct, co-dimensional rigid body simulation using unprocessed lidar data.

\section{Method}\label{sec:method}

Given an input geometry $\Omega$ embedded in $\mathbb{R}^3$, the unsigned distance field is defined as 

\begin{equation}\label{eq:dist}
d\left(\bq \right) = \min_{\Omega} d\left(\mathbf{x}, \bq \right),
\end{equation} where $\mathbf{x}\in\Omega$ is a point on the input geometry and $\bq\in\mathbb{R}^3$ is an arbitrary query point.

In our case, the input geometry is represented as a co-dimensional data mesh, $\Md = (V,F)$.
Here $V$ is a set of vertices $\bv_i$ in $\mathbb{R}^3$, and $F$ is a set of primitives.
In this paper, we deal exclusively with points, edges, and triangles, so we use the terms ``simplex'' and ``primitive'' interchangeably.
A simplex $f_i$ consists of a tuple of indices into $V$.
For example, a triangle would be represented as $f_i = (i_0, i_1, i_2)$ where each $i_k$ indexes $V$.
We also need to reference the \textit{faces} of a simplex, which we define as any non-empty subset of a simplex; for example, edges and points are faces of triangles.
Based on this definition, edges that are faces of triangles can be constructed from triangle indices using, e.g., $i_{01} = (i_0, i_1)$.
Lastly, we denote the dimension of a simplex as $n(f_i)$ and the volume of a simplex as $|f_i|$ (which is 1 for points).
$F$ can be omitted for point clouds, but we will use it throughout to keep the notation consistent.

With respect to this discretized input, the distance field computation can be reframed as finding the minimum distance between the query point and
all constituent simplicies of $\Md$:

\begin{equation}\label{eq:mindist}
    d(\Md,\bq) = \min_{i} d(f_i, \bq).
\end{equation}

\setlength{\columnsep}{0.7em}
\setlength{\intextsep}{0em}
\begin{wrapfigure}{r}{0.6\columnwidth}
  \centering
  \includegraphics[width=0.6\columnwidth]{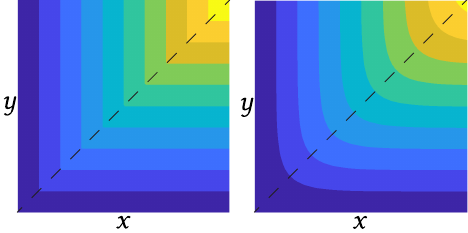}
\end{wrapfigure}

% NOTE: Exact vs smooth figure is placed in related.tex for formatting.
However, $\min$ is only $C^0$, and has discontinuities along the medial axes of the data mesh (see inset, left, for a plot of $\min(x,y)$, with the gradient discontinuity $x=y$ shown as a dashed line).
Furthermore, exact distances are a poor representation of co-dimensional geometric representations like point clouds that only loosely approximate the true underlying (volumetric) geometry.
The distance isosurfaces to a point cloud have poor gradients at small offsets and amplify sample noise at large offsets (Fig.~\ref{fig:exact_smooth}).
Our goal is to tackle both problems by designing a distance function that is at least $C^1$ differentiable and produces a smooth isosurface approximating the underlying geometry.

We begin by converting the true distance to a smooth distance via an application of the LogSumExp smooth minimum function which yields
\begin{equation}\label{eq:lse}
    \hat{d}(\Md, \bq) = -\frac{1}{\alpha} \log \left( \sum_{f_i \in F} w_i(\bq) \exp (-\alpha d_i) \right),
  \end{equation} where $d_i = d(f_i, \bq)$, $\alpha$ controls the accuracy and smoothness of the approximation and $w_i$ are influence weights for each simplex, which will be discussed in more detail in Section~\ref{sec:weights}.
(See inset, right, for a plot of $-0.1\log(\exp(-10x)+\exp(-10y))$.)

Importantly, this function is differentiable  for all finite values of $\alpha$ with gradient (with respect to $\bq$)
\begin{equation}\label{eq:dlse}
  \nabla \hat{d}(\Md, \bq) = \frac{\sum_{f_i \in F} w_i(\bq) \exp(-\alpha d_i)\nabla d_i - \frac{1}{\alpha}\exp(-\alpha d_i) \nabla w_i(\bq)}{\sum_{f_i \in F} w_i(\bq) \exp (-\alpha d_i)},
\end{equation}

and is guaranteed to underestimate the true distance to the input mesh when every $w_i \ge 1$, with error bounded by $\frac{\log(A|F|)}{\alpha}$ where $A$ is the maximum value of all $w_i$ (Appendix~\ref{sec:underestimate_proof}).
The differentiability of LogSumExp allows us to preserve the underlying smoothness of the distance functions, which are at least $C^1$ almost everywhere (Appendix~\ref{sec:exact_dist}),
and the underestimate (or conservative) property means that our smooth distance will alert us to $\bq$ crossing the input surface before it happens.
This conservative property is crucial for applications such as collision detection since it ensures that
maintaining separation with regards to the smooth distance is sufficient to maintain separation between input shapes~(Fig.~\ref{fig:noncon}).
Ergo, the output of our method will remain usable if downstream applications require the underlying geometric representations to be separated.

An additional nice property of LogSumExp is its accuracy: not only is its error merely logarithmic in mesh size, but as $\alpha$ increases, distances become more accurate.

\begin{figure}
  \includegraphics[width=\columnwidth]{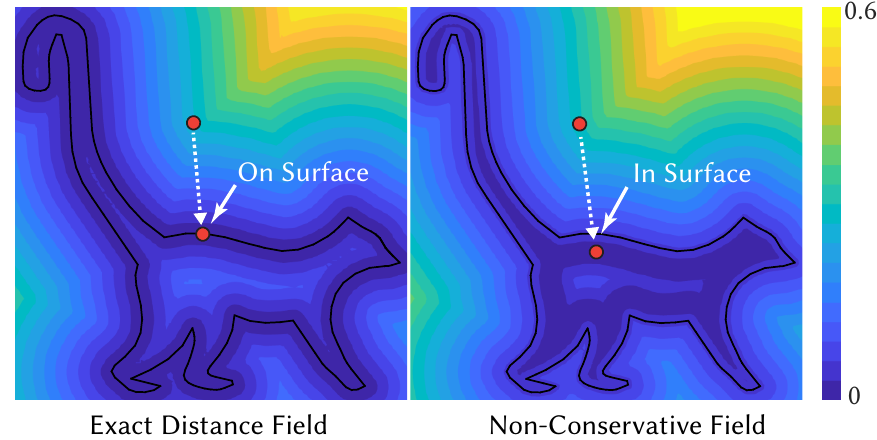}
  \caption{Points (red) starting outside a shape can be prevented from crossing into it by ensuring their trajectories (dashed line) never cross the $0$ isosurface of an unsigned distance field. 
  Non-conservative estimates break this property, allowing interpenetration of the underlying  geometry.  Such estimates are unusable if downstream algorithmic stages require the geometries to be intersection-free.}
  \label{fig:noncon}
\end{figure}

\subsection{Smooth Distance to a Single Query Point}\label{sec:dist}
As a didactic example let us apply Eq.~\ref{eq:lse} to a point cloud which we do by setting $d_i = \lVert \bq - \bv_i \rVert$ and $w_i = 1$. 
We can directly observe the smoothing effect of $\alpha$ (Fig.~\ref{fig:underestimate}), which can be used to close point clouds.
Decreasing $\alpha$ produces progressively smoother surface approximations, and surfaces produced with smaller $\alpha$ values nest those produced with higher $\alpha$'s.
This nesting is a consequence of the conservative behaviour of the LogSumExp formula.

\begin{figure}
  \includegraphics[width=\columnwidth]{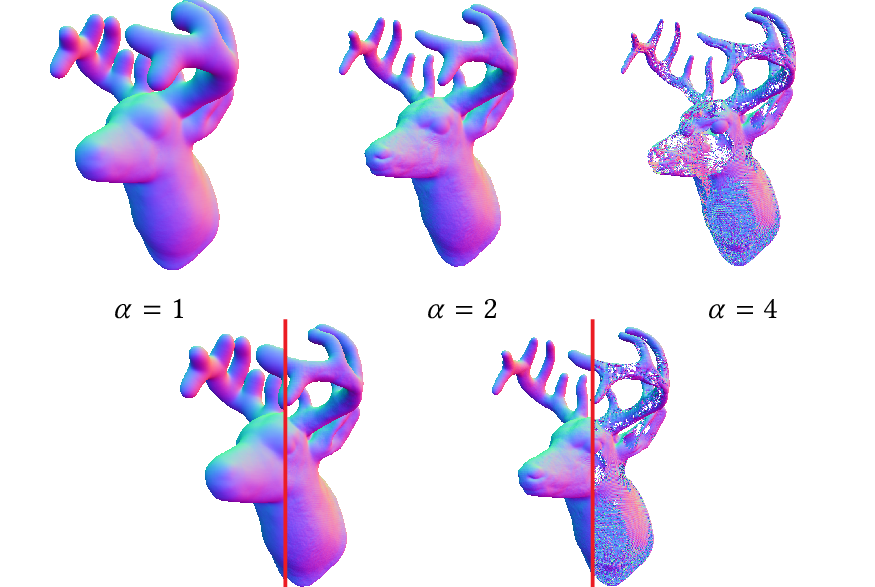}
  \caption{At low values of $\alpha$, the sharp antlers and face of this deer point cloud become puffy and smooth (top left). Increasing $\alpha$ resolves features in the face more clearly (top middle), and the individual points become visible at high $\alpha$ (top right). The bottom row shows the transition between different $\alpha$ values --- higher-$\alpha$ surfaces are contained in lower-$\alpha$ surfaces.}
  \label{fig:underestimate}
\end{figure}

All of this taken together means that the naive LogSumExp works well for point cloud geometry. 

\subsection{Smooth Distances to Co-Dimensional Geometry}
A natural extension of Eq.~\ref{eq:lse} to edge and triangle meshes is to replace the discrete sum over points with a continuous integral (see, e.g.,~\cite{Sethi2012}) over the surface:
\begin{equation*}
    \hat{d}(f_i, \bq) = -\frac{1}{\alpha} \log \left( \int_{f_i} \exp(-\alpha \lVert \bx - \bq \rVert) d\bx \right).
\end{equation*}
When discretized, this becomes equivalent to applying Eq.~\ref{eq:lse} to quadrature points on the mesh, while the $w_i$ become the quadrature weights.

While simple, this formulation will unfortunately break the important conservative property of the LogSumExp function because the quadrature weights will, in general,
not be greater than or equal to 1.
Fig.~\ref{fig:integration} shows examples of the overestimation errors introduced by this approach. 

\begin{figure}
    \includegraphics[width=\columnwidth]{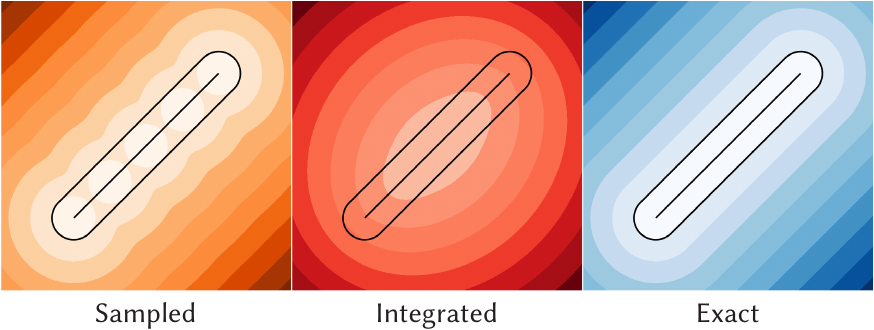}
    \caption{Computing edge distances with LogSumExp and sampled point quadrature (left) creates holes in the isosurface at high $\alpha$. Integrating over the edge (middle) using 5th order Gaussian quadrature can remove the holes but will overestimate distance at low $\alpha$. Computing exact distances (right) produces the correct isosurface. The edge and a small offset surface are shown in each, where the underestimate property requires that only the first color interval should be contained in the offset region.}
    \label{fig:integration}
\end{figure}

Rather we must return to Eq.~\ref{eq:lse} and compute the respective $d_i$'s to the constituent mesh triangles and edges exactly. 
This can be accomplished via efficient quadratic programming which we detail in Appendix~\ref{sec:exact_dist}.
However, using unit valued weights, as we did for points, produces bulging artifacts where primitives connect (Fig.~\ref{fig:weights}).
What remains is to compute weight functions that mitigate these effects while simultaneously satisfying our $\ge 1$ constraint.

\begin{figure}
  \includegraphics[width=\columnwidth]{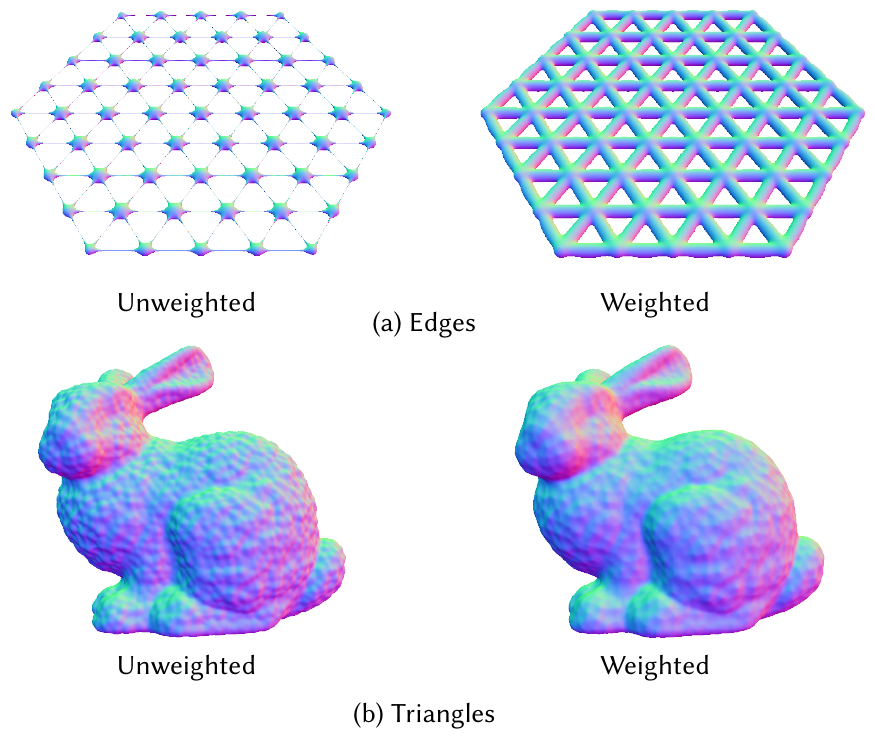}
  \caption{Distances tend to concentrate where primitives overlap, creating thin edges (top left) and bumps in the surface (bottom left). Weight functions help mitigate these effects (right), for both edge meshes (a) and triangle meshes (b).}
  \label{fig:weights}
\end{figure}

\subsection{Weight Functions}\label{sec:weights}
First, we must understand why these bulging artifacts occur.
If the closest point on $\Md$ to $\bq$ is on a simplex face, then every simplex containing that face will have the same closest point, resulting in a more severe underestimate of the true distance than usual.
While this phenomenon occurs any time $\bq$ is on the medial axis of $\Md$ (the cause of the logarithmic error term), separately computing the distance to each primitive effectively creates an artifical medial axis where there is only one true closest point that is contained in multiple primitives.
\citet{Panetta2017} observe the same concentration artifacts when using LogSumExp on edge meshes, but they propose fixing it using weighted blending, where the weight calculation relies on knowing the convex hull of the edge mesh neighborhood and blending between smooth and exact distance fields.
Unfortunately, it is unclear how to extend this to triangles or apply acceleration schemes like Barnes-Hut.
Instead, we propose a different weighting scheme which is fast, local and can be easily adapted to triangles.

We center the design for our per-primitive weight functions $w_i$ around high-accuracy use cases, which correspond to high values of $\alpha$.
The weights must be spatially varying to counteract the local concentration artifacts, and for simplicity, the function will be defined in terms of the barycentric coordinates of the closest point projection of $\bq$ onto $f_i$.
The projection function is denoted by $\bpi_i(\bq)$ and the barycentric coordinates of a point $\bp$ within $f_i$ are denoted $\bphi_i(\bp)$; using these definitions, our weight function is $w_i(\bq) = w_i(\bphi_i(\bpi_i(\bq)))$.
We will first design a weight function that mitigates the bulge artifacts without regard for maintaining the conservative property, which we will denote as $\tilde{w_i}$, and then derive an appropriate global scale factor for every $\tilde{w}_i$ in $\Md$ to achieve the conservative property.

The weight functions will be defined as polynomials in terms of $\bphi_i$, so constructing the weight functions is a matter of determining the polynomial coefficients by solving a system of linear equations, determined by both point constraints and derivative constraints.
Qualitatively, the point constraints aim to both assign a low weight to the simplex boundary and assign a high weight to the simplex interior (see Appendix~\ref{sec:weight_app} for more details).
The derivative constraints constrain the normal derivative to be 0 along the boundary, which is necessary in order to ensure $w_i$ is smooth everywhere.
To see this, we note that the gradient of $\tilde{w_i}$ with respect to $\bq$ is
\begin{equation}\label{eq:dw_new}
  \nabla \tilde{w}_i = \frac{\partial \bpi_i}{\partial \bq} \frac{\partial \bphi_i}{\partial \bpi_i} \frac{\partial \tilde{w}_i}{\partial \bphi_i},
\end{equation}
where we use the indexing convention $\left[ \frac{\partial \by}{\partial \bx} \right]_{pq} = \frac{\partial \by_q}{\partial \bx_p}$ for a generic vector function $\by(\bx)$ (i.e., gradients with respect to a scalar function are column vectors).
$\bpi_i$ is in fact a $C^0$ function, with a derivative discontinuity on the boundary.
When $\bpi_i(\bq)$ is in the interior of $f_i$, the gradient $\frac{\partial \bpi_i}{\partial \bq} \in \mathbb{R}^{3 \times 3}$ is the identity matrix, but when it is on the boundary, the gradient has a null space in the direction perpendicular to the boundary, making the normal derivative the zero vector (Fig.~\ref{fig:perp}) and creating a derivative discontinuity at the simplex boundary.
Left unchecked, this discontinuity will propagate to $w_i$ as well.
We opt to hide this discontinuity by coercing $\frac{\partial w_i}{\partial \bphi_i}$ to have zero normal derivative along the boundary while also being smooth.
See Appendix~\ref{sec:weight_app} for details on how the derivative constraints are enforced, and Fig.~\ref{fig:tri_weight} for an example weight function.

\begin{figure}
  \includegraphics[width=\columnwidth]{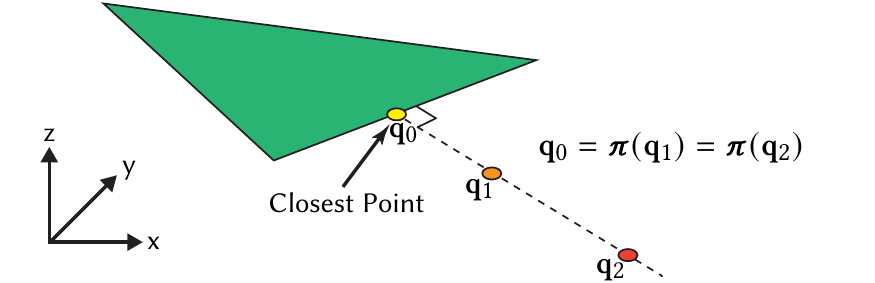}
  \caption{All points (such as $\bq_1$ and $\bq_2$) along a line extending perpendicular to a simplex boundary share a closest point $\bq_0$, and thus the normal derivative of the closest point projection $\bpi$ is the zero vector.}
  \label{fig:perp}
\end{figure}

\begin{figure}
  \includegraphics[width=\columnwidth]{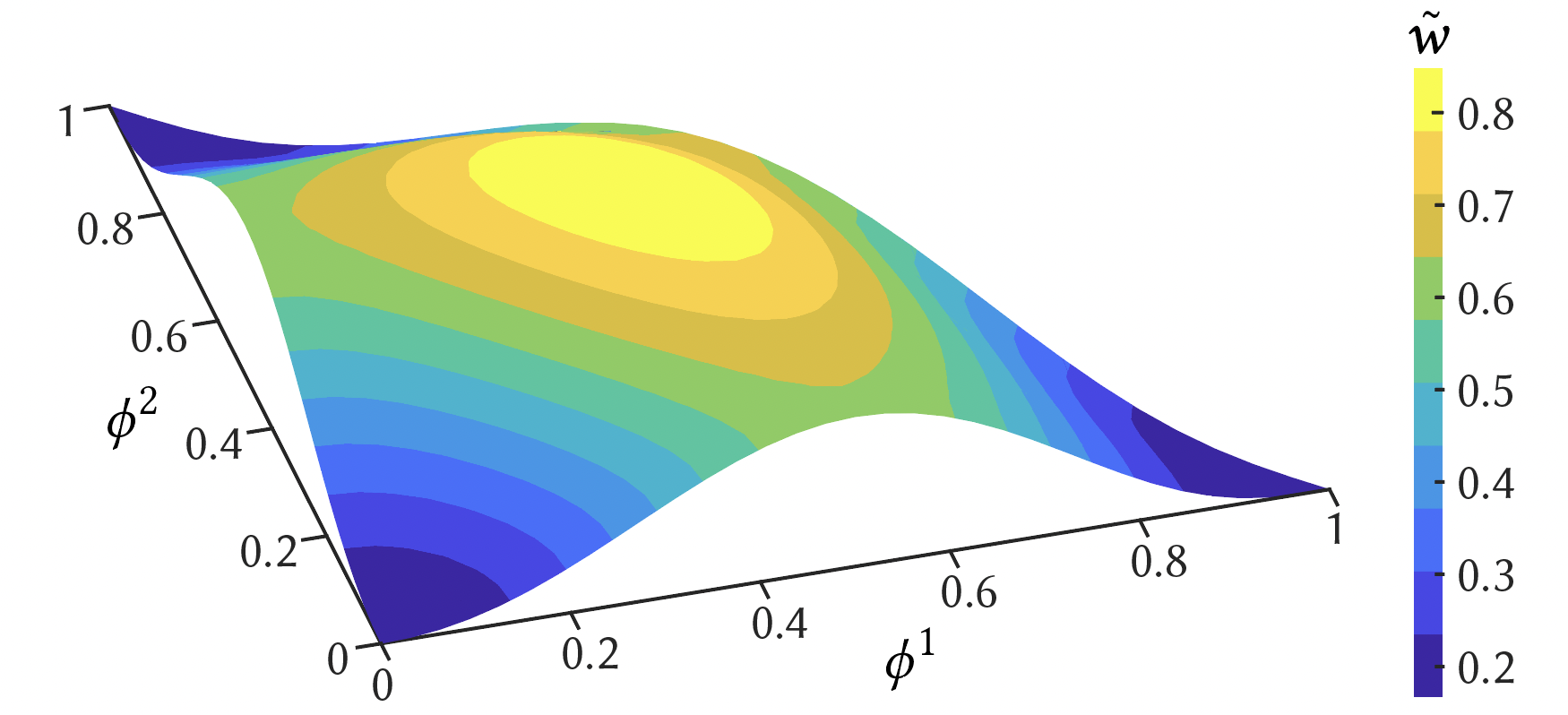}
  \caption{An example of the unscaled triangle weights, $\tilde{w}$, computed by our method for a single triangle.}
  \label{fig:tri_weight}
\end{figure}

Since the point constraints require $\tilde{w}_i$ to be less than 1 along the simplex boundary, they break the conservative property.
We rectify this by scaling $\tilde{w}_i$ by a factor $A$ associated with $\Md$ and not individual simplices, so that $A\tilde{w}_i \ge 1$ for all $f_i \in \Md$.
Although this inflates the zero isosurface, it does so uniformly, so we do not reintroduce the artifacts we wanted to remove.
This error is logarithmic in $A$ and becomes negligible at higher values of $\alpha$.

We can further refine our weight functions to improve their behaviour.
One such improvement targets the $\frac{\partial \bphi_i}{\partial \bpi_i}$ term in the gradient.
Since $\bpi_i$ is restricted to producing points within $f_i$, this term is in fact equivalent to linear shape function gradients from finite element analysis.
For triangles, our barycentric coordinate vector is $\bphi_i = [ \phi_i^1, \phi_i^2 ]^\top$, and our gradients (with respect to points in $f_i$) are $\nabla \phi_i^1 = \frac{(\bv_{i_0} - \bv_{i_2})^\bot}{2|f_i|}$ and $\nabla \phi_i^2 = \frac{(\bv_{i_1} - \bv_{i_0})^\bot}{2|f_i|}$ where $\bx^\bot$ represents a 90-degree counterclockwise rotation of $\bx$.
We see that gradients are inversely proportional to triangle area, and thus create gradient artifacts in smaller triangles (Fig.~\ref{fig:metric}).
Our goal is to control the magnitude of these gradients.

Since larger triangles have smaller gradients, a simple way to fix the gradients is to isotropically scale the space by a factor $\rho$, compute $\hat{d}$, and then scale back to the original space at the end.
This way, weight gradients are computed in the scaled space and we can directly control the magnitude of the barycentric gradients.
However, if we expand Eq.~\ref{eq:lse} using distances scaled by $\rho$, we notice that $\rho$ behaves the exact same way as $\alpha$, and so we do not actually need a new parameter.
Instead, $\alpha$ itself can be interpreted as a uniform scale factor in space, and controlling accuracy with $\alpha$ is equivalent to measuring lengths with the metric $\alpha^2I$ where $I$ is the identity matrix.
Going back to weight gradients, we now measure lengths in $\alpha$-scaled space, and we get barycentric gradients such as $\nabla \phi_i^1 = \frac{(\bv_{i_0} - \bv_{i_2})^\bot}{2\alpha|f_i|}$, which allows us to reduce the norm of the problematic gradient term (Fig.~\ref{fig:metric}).
Other quantities stay the same since they are either independent of the metric (e.g., $\nabla \bpi_i$), or their dependence on $\alpha$ and thus the metric space is explicit and already accounted for.

\begin{figure}
  \includegraphics[width=\columnwidth]{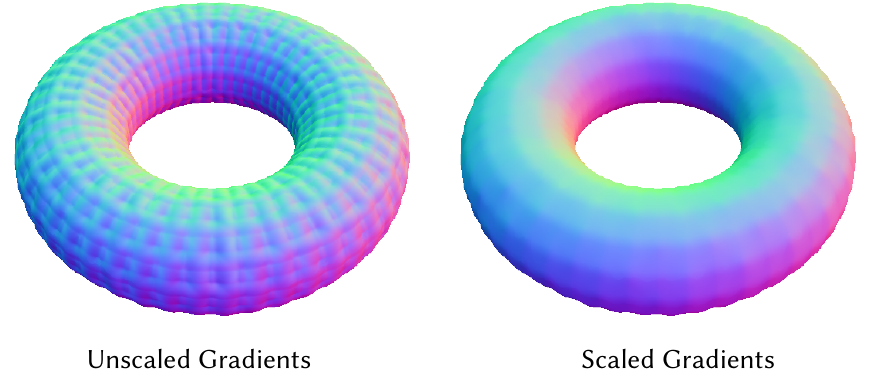}
  \caption{Small meshes produce large weight gradients which produce noticeable ridges on this torus (left); uniformly scaling the ambient space allows us to control the length of these gradients and smooth out the ridges (right).}
  \label{fig:metric}
\end{figure}

Another improvement is the behaviour of the $w_i$'s at low $\alpha$.
A key assumption in our design of $w_i$ is that $\alpha$ is sufficiently large, but since $\alpha$ is a controllable user parameter, this assumption can sometimes be violated, and the weights can overcompensate for concentration.
In these cases $w_i$ needs to be further modified to avoid artifacts (Fig.~\ref{fig:attenuate}).
If $\alpha$ surpasses some threshold $\alpha_U$ (which at a high level represents a upper bound on $\alpha$ --- see Section~\ref{sec:params} for more information on how it is selected), then we can use $w_i$ as-is; if $\alpha < \alpha_U$, then we must flatten or attenuate $w_i$.
Experimentally we observe that using a scale factor $S = \frac{\alpha}{\max(\alpha, \alpha_U)}$ to define an \textit{attenuated} weight function $w_i(\bq) = (A \tilde{w}_i(\bq))^S$, with gradient $\nabla w_i(\bq) = S(A\tilde{w}_i(\bq))^{S-1} \nabla \tilde{w}_i(\bq)$, helps minimize these issues.
Note that this heuristic does not eliminate the issue for all $\alpha$ but does mitigate it sufficiently for all ranges of $\alpha$ used in this paper.

\begin{figure}
  \includegraphics[width=\columnwidth]{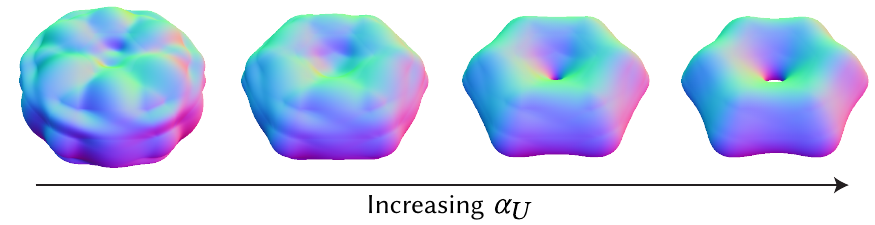}
  \caption{At a low value of $\alpha$, the weight function overcompensates and produces creases and other spurious artifacts in this very blobby hexagonal torus; adjusting $\alpha_U$ makes it possible to blend between weighted and unweighted distances to mitigate these artifacts.}
  \label{fig:attenuate}
\end{figure}

\subsection{Generalized Query Primitives}\label{sec:general}
We will briefly discuss how to generalize our query point into a general query primitive $g$, which can also be an edge or a triangle.
Just like with points, we can compute the distance to another simplex $f_i$ as $d(f_i, g)$ using a positive semidefinite quadratic program with linear constraints (see Appendix~\ref{sec:exact_dist}).
We now also obtain barycentric coordinates for $g$ through the argmin, which we will denote as $\blambda$.
All formulas in the preceding subsections can be adapted for this generalization by simply replacing $\bq$ with $g$, and using the closest point to $\Md$ on $g$, $g(\blambda)$, when a single point is needed (e.g., in weight functions).
These changes imply that $\nabla \hat{d}$ is now with respect to an entire query primitive, but this simply describes a rigid translation of $g$.
Taking gradients with respect to entire primitives also ensures that, even in configurations with multiple closest point pairs such as parallel edges, the gradient is the same for each closest point pair and thus well-defined.

\subsection{Barnes-Hut Approximation}\label{sec:bh}
Although our distance function is smooth and easy to differentiate, it requires evaluation of a distance between every pair of primitives in the most general case.
However, its use of exponentially decaying functions enables the application of the well-known Barnes-Hut approximation~\cite{Barnes1986} to accelerate evaluation.
Barnes-Hut uses a \textit{far field approximation} to cluster groups of far-away primitives together (typically using a spatial subdivision data structure like an octree or bounding volume hierarchy) and treat them as a single point.
The approximation is characterized by the centers of each bounding region, and a user parameter $\beta$ which controls where the far field approximation is employed --- see Appendix~\ref{sec:farfield} for details.
At a high level, lower $\beta$ is more accurate, and $\beta = 0$ results in an exact evaluation.

Placing the far-field expansion at the center of mass reduces overall error~\cite{Barnes1986} but it can possibly overestimate the true distance, breaking the conservative nature of the LogSumExp smooth distance.
A minor modification can reinstate this property -- placing the expansion center on the closest point of the cluster region to the query primitive $g$, rather than at the center of mass.
When using an octree or BVH, we pick the closest point on the bounding box to the query point.
This may increase the error relative to using the center of mass, but it guarantees that the Barnes-Hut estimate only underestimates the exact smooth distance (Fig.~\ref{fig:approx}).
Note that this choice does not affect the gradient if the bounding region is convex (which is the case for bounding box hierarchies and octrees), since rigidly translating $g$ farther away along the gradient direction does not change the closest point.

\begin{figure}
  \includegraphics[width=\columnwidth]{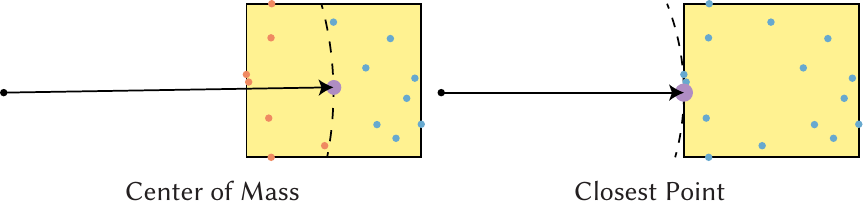}
  \caption{The center of the Barnes-Hut approximation (purple) has a significant effect on the sign of the error. The center of mass can potentially be farther away than some data points (orange), which can lead to an overestimate of the true distance. On the other hand, the closest point on the bounding box is guaranteed to be closer than every data point (blue), so the final result will always be an underestimate.}
  \label{fig:approx}
\end{figure}

Finding the closest point on a box to a query point is simple, but finding the closest point on a box to a query edge or triangle is more complex, and requires solving a quadratic program.
However, this is relatively expensive for what is meant to be a fast check to help reduce computation time, so we instead use an approximation of the problem.
Viewing a box as the intersection of 6 halfspaces, we identify which halfspaces the primitive lies completely outside of.
If there are 3 such halfspaces, they identify a corner of the box that is closest to the primitive; if there are 2 halfspaces, they identify an edge of the box; and if there is 1 halfspace, it identifies a face of the box.
If there are no halfspaces that satisfy this criteria, then we simply return the distance to the box as 0 and force the traversal to visit its children.
Although this test somewhat inhibits the Barnes-Hut approximation's ability to group primitives together, the computational savings per visited node more than make up for it --- we noticed over a $5\times$ improvement in performance in our experiments.

Psueodocode summarizing our method so far is given in Algorithm~\ref{alg:contrib} and Algorithm~\ref{alg:mindist}.

\begin{algorithm}
  \caption{Collecting contributions for $\hat{d}(\Md, g)$ \hspace{4em} \textsc{collectContributions}}\label{alg:contrib}
  \SetAlgoLined
  \SetKwInOut{Input}{Inputs}
  \SetKwInOut{Output}{Outputs}
  \Input{Data mesh $\Md$, BVH node $B$, query primitive $g$, parameters $\alpha$, $\alpha_U$, $\beta$}
  \Output{Sum of exponentials $c$ and sum of weighted distance gradients $\nabla c$}
  \tcp{$B$ has children $B.l$ and $B.r$}
  \BlankLine
  \uIf{$\textsc{BarnesHutCondition}(\Md, B, g, \beta)$ \tcp*{App.~\ref{sec:farfield}, Sec.~\ref{sec:bh}}}{
    Return far field expansion \tcp*{App.~\ref{sec:farfield}}
  }
  \uElseIf{$B$ is a leaf containing $f_i$}{
    $(d_i, \nabla d_i, \bphi_i, \blambda) \leftarrow d(f_i, g)$ \tcp*{Sec.~\ref{sec:general}, App.~\ref{sec:exact_dist}}
    $(w_i, \nabla w_i) \leftarrow \textsc{weightFn}(f_i, \bphi_i, g, \blambda, \alpha, \alpha_U)$ \tcp*{Sec.~\ref{sec:weights}}
    $c \leftarrow w_i \exp(-\alpha d_i)$ \tcp*{Eq.~\ref{eq:lse}}
    $\nabla c \leftarrow c \nabla d_i - \frac{1}{\alpha} \exp(-\alpha d_i) \nabla w_i$ \tcp*{Eq.~\ref{eq:dlse}}
    Return $(c, \nabla c)$\;
  }
  \Else{
    $(c_l, \nabla c_l) \leftarrow \textsc{collectContributions}(\Md, B.l, g, \alpha, \alpha_U, \beta)$\;
    $(c_r, \nabla c_r) \leftarrow \textsc{collectContributions}(\Md, B.r, g, \alpha, \alpha_U, \beta)$\;
    Return $(c_l + c_r, \nabla c_l + \nabla c_r)$\;
  }
\end{algorithm}

\begin{algorithm}
  \caption{Computing $\hat{d}(\Md, g)$ \textsc{smoothMinDist}}\label{alg:mindist}
  \SetAlgoLined
  \SetKwInOut{Input}{Inputs}
  \SetKwInOut{Output}{Outputs}
  \Input{Data mesh $\Md$, BVH node $B$, query primitive $g$, parameters $\alpha$, $\alpha_U$, $\beta$}
  \Output{Smooth min distance $\hat{d}$ and gradient $\nabla \hat{d}$}
  \BlankLine
  $(c, \nabla c) \leftarrow \textsc{collectContributions}(\Md, B, g, \alpha, \alpha_U, \beta)$\;
  $\hat{d} \leftarrow -\frac{1}{\alpha} \log c$\;
  $\nabla \hat{d} \leftarrow \frac{\nabla c}{c + \epsilon}$ \tcp*{Avoid divide-by-zero}
\end{algorithm}

\subsection{Smooth Distance to a Query Mesh}

Now that we are equipped with an efficient method for computing $\hat{d}(\Md, g)$, we can combine these distances using LogSumExp to obtain a smooth distance between $\Md$ and a \textit{query mesh} $\Mq = (\bV, \bF)$:
\begin{equation}\label{eq:lse_full}
  \hat{d}(\Md, \Mq) = -\frac{1}{\alpha_q} \log \left( \sum_{g_j \in \bF} \exp \left( -\alpha_q \hat{d}(\Md, g_j) \right) \right),
\end{equation}
where we have introduced a new accuracy-controlling parameter $\alpha_q$ that is independent of the inner $\alpha$.
In an exact evaluation ($\beta = 0$), this strongly resembles the LogSumExp of all the pairwise distances between each $f_i$ and $g_j$.

A subtle issue with the current formulation is that the distance gradients $\nabla \hat{d}(\Md, g_j)$ are with respect to $g_j$ as a whole, but we need per-vertex gradients as those are the true degrees of freedom.
One simple way to do this is to split $\exp(-\alpha_q \hat{d}(\Md, g_j))$ equally between the vertices of $g_j$, and rewriting the summation over vertices and one-ring neighbourhoods gives us:
\begin{equation}\label{eq:lse_full2}
  \hat{d}(\Md, \Mq) = -\frac{1}{\alpha_q} \log \left( \sum_{\bq_k \in \bV} \sum_{g_j \in \mathcal{N}_k} \frac{1}{n(g_j)+1} \exp \left( -\alpha_q \hat{d}(\Md, g_j) \right) \right),
\end{equation}
where $\mathcal{N}_k$ is the set of one-ring neighbours of $\bq_k$.
Then, the gradient with respect to query vertex $\bq_k$ is
\begin{equation}\label{eq:dlse_full}
  \nabla_k \hat{d}(\Md, \Mq) = \frac{\sum_{g_j \in \mathcal{N}_k} \frac{1}{n(g_j)+1} \exp \left( -\alpha_q \hat{d}(\Md, g_j) \right) \nabla \hat{d}(\Md, g_j)}{\sum_{g_j \in \bF} \exp \left( -\alpha_q \hat{d}(\Md, g_j) \right)}.
\end{equation}
Once again, the convexity of the query and data primitives means that we can interpret the $\nabla \hat{d}(\Md, g_j)$ as a rigid translation of $g_j$ that affects all its vertices equally.

The summation in Eq.~\ref{eq:lse_full} can be easily parallelized, and the gradient computation requires only a small amount of serialization at the end to redistribute gradients onto vertices.

\section{Results}

\subsection{Implementation}
We implemented our method using C++ with Eigen~\cite{eigen} and \textsc{libigl}~\cite{libigl}.
We designed the implementation so that it could be ported onto the GPU, and to this end, we implemented the BVH traversal algorithm outlined in Algorithm~\ref{alg:contrib} using the stackless method of \citet{Hapala2011}.
Since GPUs exhibit poor performance for double-precision floating point, most of our computations (particularly vector arithmetic) are conducted in single-precision, while exponential sums are tracked using double-precision to increase the range of usable $\alpha$ values.
Primitive distances were hand coded if feasible, and were otherwise implemented using a null-space quadratic program solver written using Eigen (edge-triangle and triangle-triangle distances).
An important aspect of these distance computations is robustness, which becomes particularly important because distances are computed using single-precision floats, and errors in the distance can result in breaking the conservative property.
The hand-coded distance functions made extensive use of an algorithm by Kahan which employs the fused-multiply-add instruction to reduce cancellation error~\cite{Kahan2004}, while the quadratic program solver leveraged Eigen's numerically stable algorithms.

\subsection{Sphere Tracing}\label{sec:spheretrace}
Sphere tracing is a method to render signed (and unsigned) distance functions~\cite{Hart1996}.
We can use sphere tracing to visualize the zero isosurface of $\hat{d}$ with $\Mq$ as a single point, which we have done throughout the paper for demonstrative purposes.

\subsection{Parameter Analysis}\label{sec:params}
To demonstrate the effectiveness of Barnes-Hut, we conduct an ablation study on $\beta$.
In order for Barnes-Hut to be useful in approximating a constraint function, it primarily needs to be accurate near the zero isosurface, as that is where it is evaluated in constrained optimization problems (e.g., rigid body contact constraints are only evaluated when there is a potential collision).
The approximation becomes increasingly inaccurate farther away from the surface, but since we are only concerned with the zero isosurface, we use sphere tracing as a sampling technique.
Using the Stanford bunny mesh's vertices, edges, and triangles, with $\alpha=200$, $\alpha_U=1200$, and $\beta$ between 0 and 1, we measure the amount of time taken to render a $512 \times 512$ image, as well as the distance along each ray between the approximation's estimate of the isosurface and the actual isosurface ($\beta=0$).
The results are shown in Fig.~\ref{fig:betatest}, using 4 threads on a 2015 MacBook Pro.
We see that running time decreases by an order of magnitude even for $\beta=0.2$ in all three cases, and all renders take less than 10s at $\beta=0.5$.
The error relative to the bounding box diagonal also remains below 4\% for these values of $\beta$.
As a result, we use $\beta=0.5$ in all of our examples in the paper unless otherwise stated.
We find that $\beta$ can be increased with low error for higher values of $\alpha$ (or equivalently, meshes with lower sampling density), but this is not necessary to obtain significant speedups.

\begin{figure}
  \begin{subfigure}{\columnwidth}
    \includegraphics[width=\columnwidth]{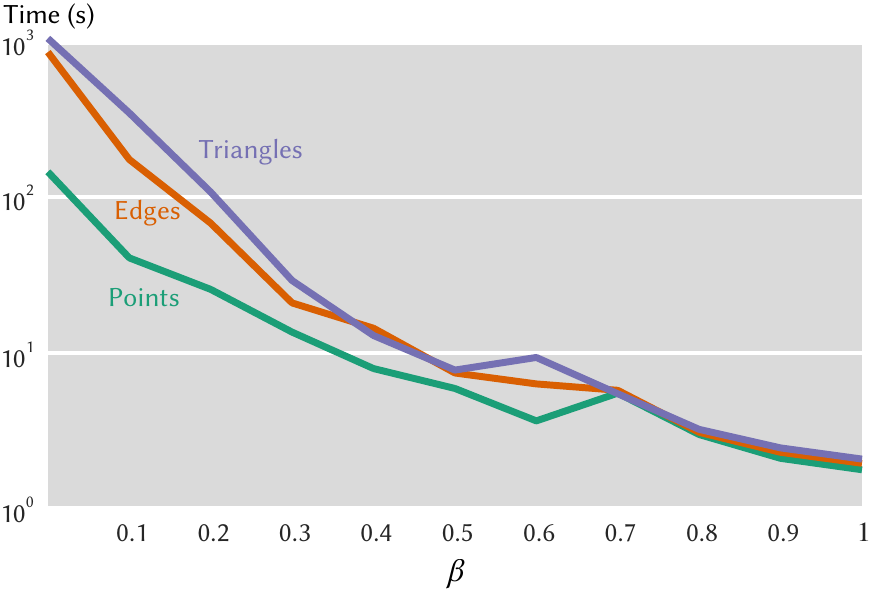}
    \label{fig:betatime}
  \end{subfigure}
  \begin{subfigure}{\columnwidth}
    \includegraphics[width=\columnwidth]{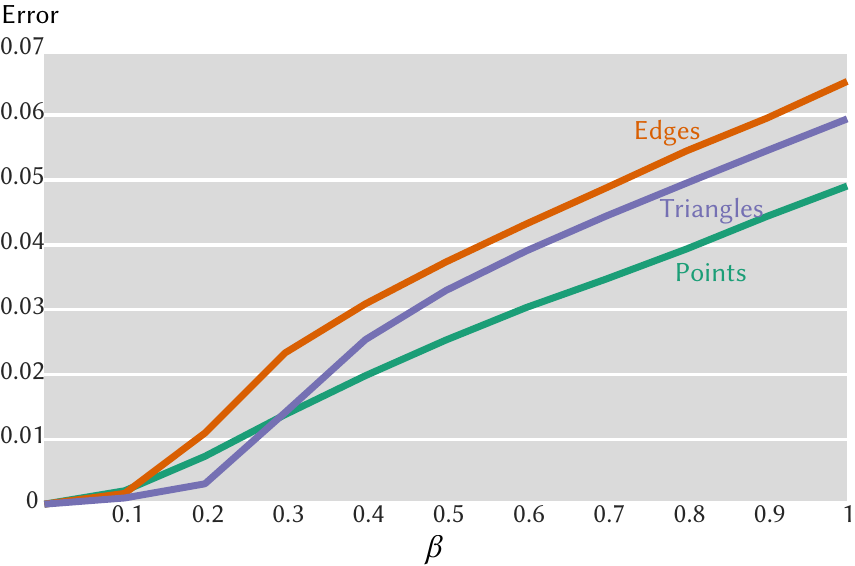}
    \label{fig:betaerr}
  \end{subfigure}
  \caption{Sphere tracing the Stanford bunny's smooth distance function over its vertices (green), edges (orange), and triangles (purple) obtains order-of-magnitude speedups (top) even for small $\beta$, while achieving low isosurface approximation error relative to the bounding box diagonal (bottom).}
  \label{fig:betatest}
\end{figure}

Due to Barnes-Hut relying on a switch between a near field and far field approximation, the smooth distance field can exhibit discontinuities.
In order to quantify the effect of these discontinuities, we perform a sensitivity analysis on $\beta$ at points where the evaluation switches between a near field and far field expansion, since we expect a change in $\beta$ to also change the leaves visited in the evaluation.
Using the city lidar point cloud from Fig.~\ref{fig:bike_city} and $\beta=0.5$, we generate sample points by uniformly sampling points on the top and bottom of the bounding box, and sphere tracing rays along the vertical coordinate axis (z-axis in this case) until the ray either misses the surface or has a smooth distance of 0.1 or less.
Then, with these sample points on the 0.1 isosurface, we trace small rays in the same vertical direction, and whenever the evaluation switches from a far field to near field expansion in a bounding box (i.e., traverses inside a new bounding box), we also measure the smooth distance using $\beta=0.5001$ and compute the discrepancy in smooth distance and the cosine between gradient vectors.
We collect 259 field switch points where $\beta=0.5001$ traverses less of the BVH than $\beta=0.5$, and observe that the mean smooth distance discrepancy is $7.6 \times 10^{-5}$ and the maximum discrepancy is $6.6 \times 10^{-3}$; for gradients, the mean cosine is 1 and the minimum cosine is 0.9974.
These results are sufficient for our rigid body simulations with geometry orders of magnitude larger than these discrepancies, though for applications with more demanding smoothness requirements, smooth Barnes-Hut approximations are a potential area for future work.

The remaining parameters $\alpha$ and $\alpha_U$ are associated with the geometry of $\Md$.
In cases where we want a function resembling our underlying geometry, we want to select an $\alpha$ that is high enough to produce a good approximation of the surface, but low enough to prevent numerical problems.
Since we can interpret $\alpha$ as a metric, we can select $\alpha$ based on the density of our geometry.
One heuristic we found to be a useful starting point for edge and triangle meshes is to set $\alpha$ to be the reciprocal of the minimum edge length.
For points, we set $\alpha$ to be 100 times the length of the bounding box diagonal.
These results can be refined by sphere tracing the zero isosurface and tweaking the results until the surface looks satisfactory.
$\alpha_U$ is essentially an upper bound on $\alpha$ and can be determined using this same procedure, which allows it to be used in low-$\alpha$ scenarios as well.
$\alpha_q$ is very similar to $\alpha$ but is related to the geometry of $\Mq$, and can be selected using the same process.
A fully automated solution for selecting these parameters is left as future work.

\subsection{Performance Benchmark}\label{sec:benchmark}

We performed a large-scale performance evaluation of our method on the Thingi10K dataset~\cite{Thingi10K}.
For each model, we created 3 data meshes: all of the model's vertices $V$, all of its edges $E$, and all of its triangles $F$.
We then scaled and translated each of these meshes so that their bounding boxes were centered at $[0.5, 0.5, 0.5]^\top$, with a bounding box diagonal of 0.5.
Then, we evaluated $\hat{d}$ from each mesh to each of the voxel centers of a $100 \times 100 \times 100$ grid between $[0, 0, 0]^\top$ and $[1, 1, 1]^\top$, in parallel using 16 threads.
The results are reported in Fig.~\ref{fig:benchmark}.
We can see that performance is proportional to the number of visited leaves (i.e., the number of primitive distance calculations and far field expansions employed), and the percentage of visited leaves drops significantly as meshes increase in size.
These results show that our method is quite scalable, and is very good at handling large meshes.
Just like in Section~\ref{sec:spheretrace}, points tend to perform much better than edges and triangles, but now we can see why this is the case.
Point clouds have lower variances than edge meshes and triangle meshes in their average leaves visited percentage --- for example, even at at 1000 edges/triangles, several meshes require the traversal to visit well over 10\% of the leaves in their BVH on average.
Also, points generally visit far fewer leaves --- point cloud tests visit at most 70 leaves, while edge meshes and triangle meshes can require over 2000 visited leaves.
As we can see from the first graph, less leaves visited corresponds to faster query times, and so points are empirically faster than the other two types of meshes.

We also ran this benchmark on a GPU --- see Appendix~\ref{sec:gpu_benchmark} for the results.

\begin{figure*}
  \includegraphics[width=\textwidth]{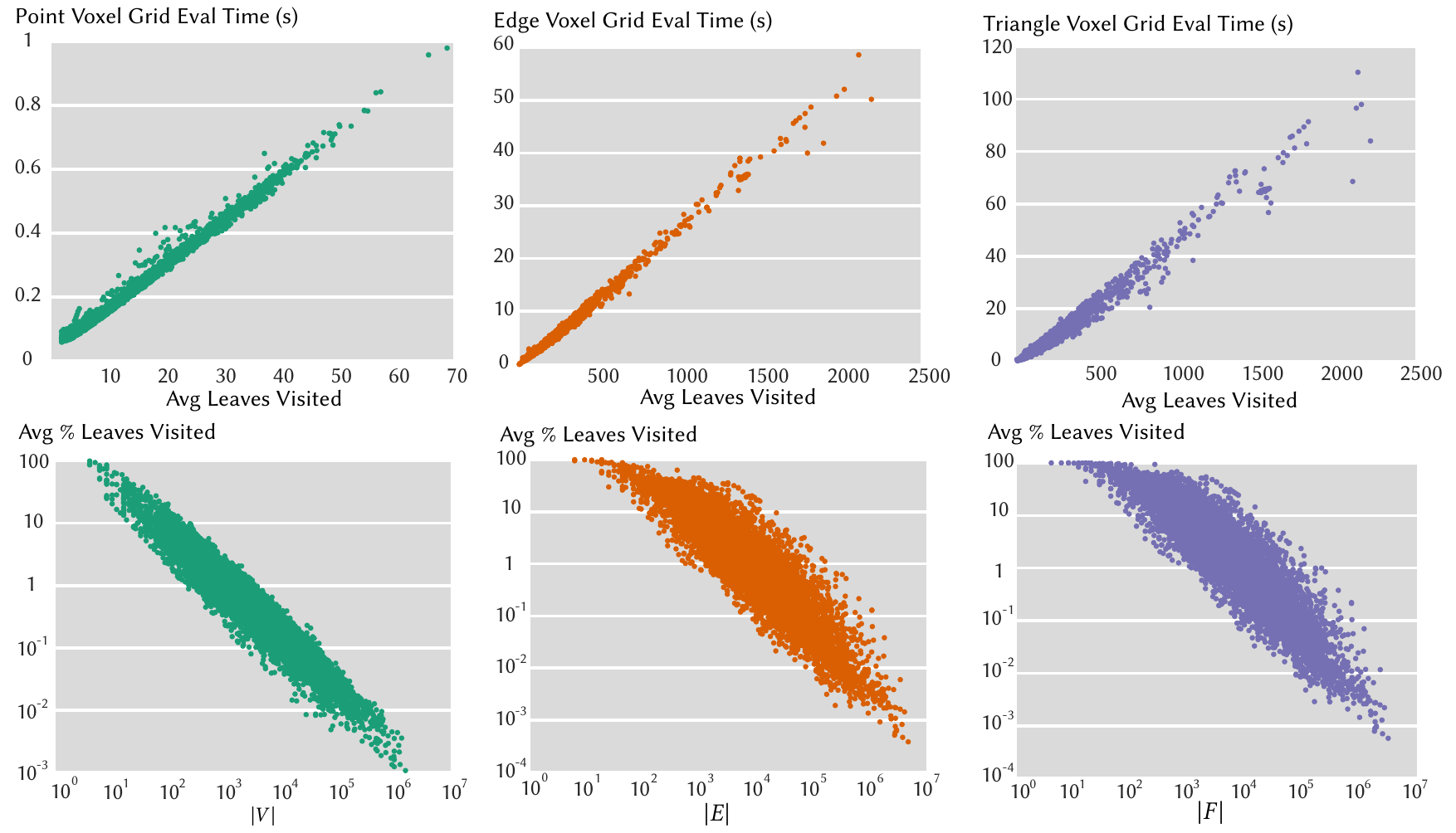}
  \caption{Results of the Thingi10K benchmark, using the vertices $V$, edges $E$, and triangles $F$ of each mesh in the dataset to evaluate queries on a $100 \times 100 \times 100$ voxel grid. The top row shows that running time is linearly proportional to the number of visited leaves for all mesh types, and the percentage of visited leaves drops sharply as mesh size increases. Furthermore, points tend to perform much better than edges and triangles.}
  \label{fig:benchmark}
\end{figure*}

\subsection{Rigid Body Collisions}\label{sec:rb}
One popular application of distance functions is in collision resolution, where they can be used as an intersection-free optimization constraint.
To focus on the distance function constraint rather than the mechanics, we simulate frictionless rigid body contact, using a variety of configurations and a variety of geometry.
Each time step of our simulation is driven by an incremental potential energy based on \citet{Ferguson:2021:RigidIPC} with a single smooth distance constraint.
For example, the optimization problem corresponding to a single object $\Mq$ colliding with one other static object $\Md$ in the scene is
\begin{equation}\label{eq:opt}
  \begin{aligned}
    \min_{\bp, \btheta} \quad & E(\bp, \btheta) \\
    s.t. \quad & \hat{d}(\Md,\Mq(\bp,\btheta)) \ge 0,
  \end{aligned}
\end{equation}
where $\bp$ and $\btheta$ are vectors in $\mathbb{R}^3$ describing the world space position and orientation of $\Mq$, respectively, and $E$ is an objective function whose unconstrained minimizer is equivalent to an implicit Euler time step in position and an exponential Euler time step in orientation.
We solved this optimization problem using a primal-dual interior-point solver~\cite{Nocedal2006} written in C++, where every iteration produces a feasible point, and we replace the Hessian block of the primal-dual system with an identity matrix (similar to gradient descent on the optimization problem's Lagrangian).
For multi-object simulations, each object in the scene has an associated $\alpha$ and $\alpha_U$ to use in smooth distance evaluations, and every pairwise smooth distance constraint is combined into a single constraint using LogSumExp with the maximum $\alpha$ among all objects in the scene.
Inertial quantities were computed using the underlying geometry.
We did not implement continuous collision detection, so we use relatively small time steps in our examples.

Since we wrap every pairwise primitive distance constraint into a single constraint, a much simpler alternative to smooth distances is the exact minimum of all pairwise distances.
However, as we discussed in Section~\ref{sec:method}, the exact minimum is $C^0$, which is undesirable (other methods that use exact distance use multiple constraints, which is better behaved than a single exact minimum constraint~\cite{Nocedal2006}).
We demonstrate its poor behaviour in two 2D examples where a point is dropped into a V-shaped bowl, and compare the results between exact minimum distance and smooth distance constraints (Fig.~\ref{fig:2d_exact_smooth}).
In both simulations, exact distances perform qualitatively worse than smooth distances, because they do not smooth the geometry and must deal with the sharp gradient change after passing the medial axis passing vertically through the base of the bowl.

\begin{figure}
  \includegraphics[width=\columnwidth]{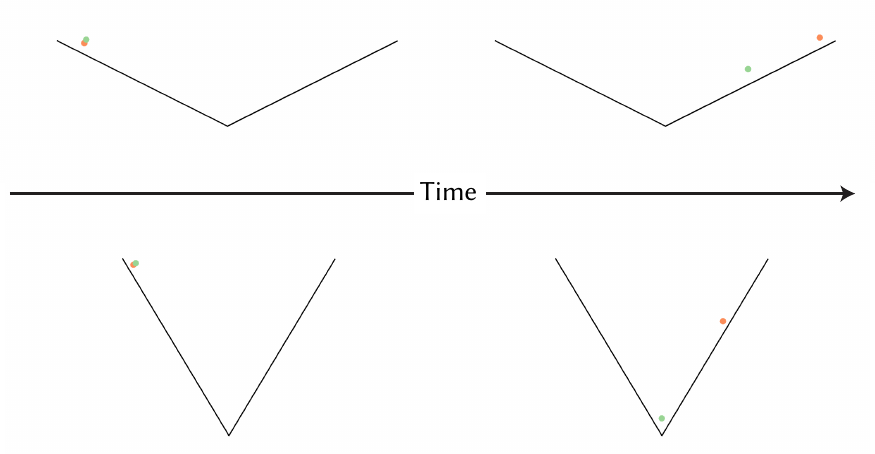}
  \caption{A comparison of exact and smooth distance constraints in rigid body simulations of a point mass dropping onto two different V-shaped bowls. In the shallow bowl (top), exact distances (green) lose a lot of kinetic energy after the sharp base, while smooth distances (orange) allow the point to roll out of the bowl. In the deep bowl (bottom), exact distances simply get stuck, while smooth distances are able to continue past the sharp base for a time.}
  \label{fig:2d_exact_smooth}
\end{figure}

We now show a variety of simulations in three dimensions using smooth distances.
First, we show two rigid bunnies colliding together and falling onto a bumpy ground plane as a simple test (Fig.~\ref{fig:bunny_bunny}).
The two bunnies bounce off each other and fall down, without intersections.
Now we depart from traditional examples and simulate co-dimensional geometry.
We can simulate the edges of a faceted icosphere falling into a net-like bowl (Fig.~\ref{fig:ico_bowl}), where it rolls to the other side and falls back into the bowl.
Both of these meshes are represented as edge meshes in the simulation, and we see again that there is no interpenetration.
Similarly, we can show a sphere roll down a slide represented as an edge mesh (Fig.~\ref{fig:sphere_slide}).
In a more complex scenario, we can mix primitive types in a mesh and still compute distances.
We can take a sphere mesh, attach some spikes represented as edges to it, and then simulate this shape falling into the net bowl (Fig.~\ref{fig:spikeball_bowl}).
We see one of the spikes hit the lip of the bowl and the spiky ball rolling up the other side of the bowl.

We can go even further and use our function to simulate cases that have not been well-defined in previous methods, like point cloud collisions.
Although recent work has simulated contact with highly co-dimensional geometry~\cite{Li2021CIPC,Ferguson:2021:RigidIPC}, the inflated isosurfaces provided by $\hat{d}$ allow us to close the surface defined by the point cloud without a surface reconstruction preprocess.
We demonstrate this by tossing a trefoil knot, represented as a piecewise linear curve, into a ring toss game represented as a point cloud (Fig.~\ref{fig:trefoil_ringtoss}).
In larger examples, we simulate an octopus triangle mesh (Fig.~\ref{fig:octopus_lidar}) and a point cloud bunny (Fig.~\ref{fig:bunny_lidar}) sliding and rolling down point cloud terrain acquired from a lidar scanner.
We also simulated a motorcycle jumping onto a lidar-acquired point cloud street (Fig.~\ref{fig:bike_city}).

\begin{figure}
  \includegraphics[width=\columnwidth]{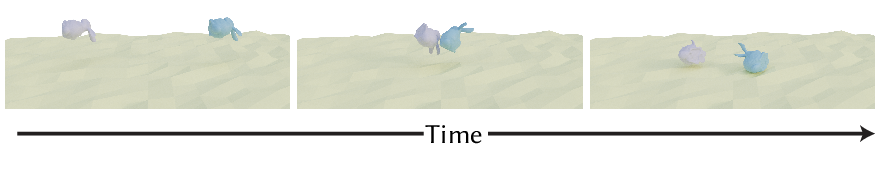}
  \caption{Two bunnies colliding and falling onto a bumpy plane, where they bounce along the surface.}
  \label{fig:bunny_bunny}
\end{figure}

\begin{figure}
  \includegraphics[width=\columnwidth]{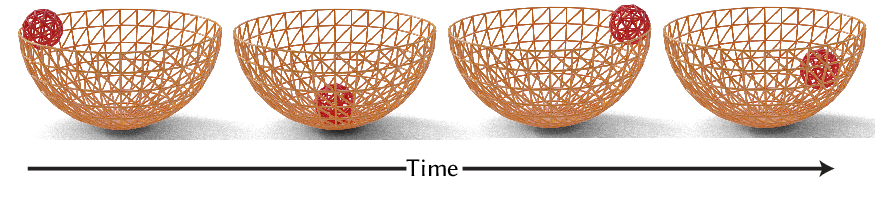}
  \caption{An icosphere edge mesh is dropped into a hemispherical bowl which is also represented as an edge mesh, where it rolls around in the bowl. Note that both meshes are only rendered as triangle meshes, and are represented in the simulation as edge meshes.}
  \label{fig:ico_bowl}
\end{figure}

\begin{figure}
  \includegraphics[width=\columnwidth]{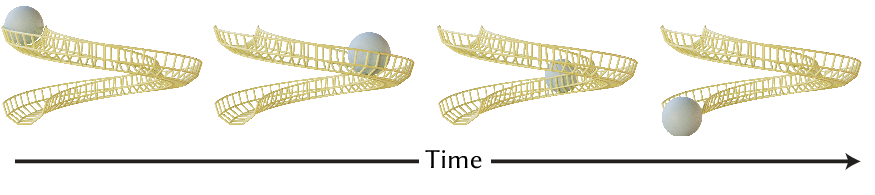}
  \caption{A triangle mesh sphere rolls down a twisting slide made of edges.}
  \label{fig:sphere_slide}
\end{figure}

\begin{figure*}
  \includegraphics[width=\textwidth]{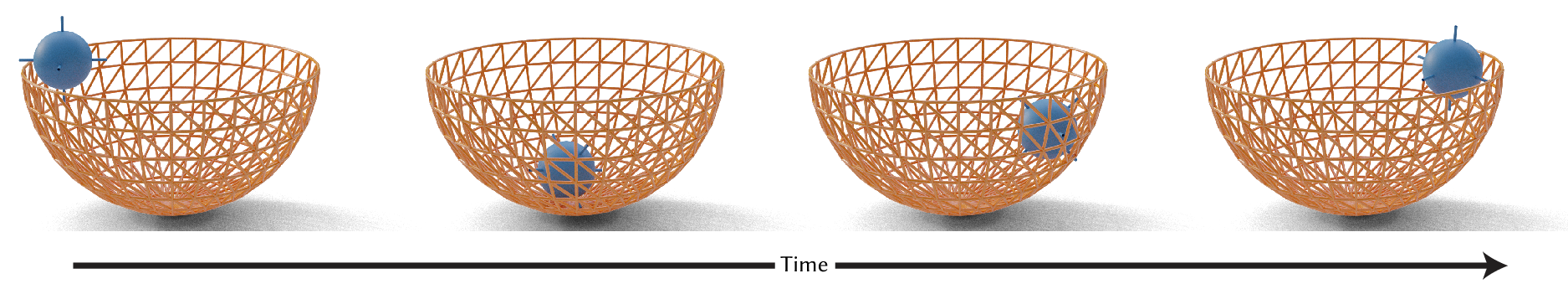}
  \caption{A spiky sphere made from a triangle mesh and edges for the spikes, falls into a bowl edge mesh and rolls up the side of the bowl, with the spikes hitting the bowl's wires. Due to low $\alpha$ values, the spikes rarely poke through the bowl.}
  \label{fig:spikeball_bowl}
\end{figure*}

\begin{figure*}
  \includegraphics[width=\textwidth]{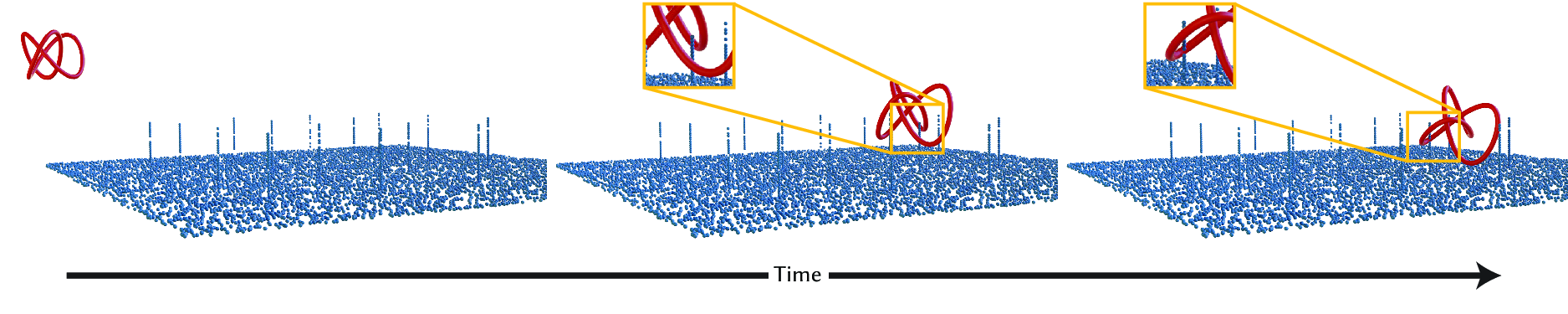}
  \caption{A trefoil knot is thrown into a ring toss game, represented by point samples, where it hits a ring spike and slides around it before resting on the base.}
  \label{fig:trefoil_ringtoss}
\end{figure*}

\begin{figure*}
  \includegraphics[width=\textwidth]{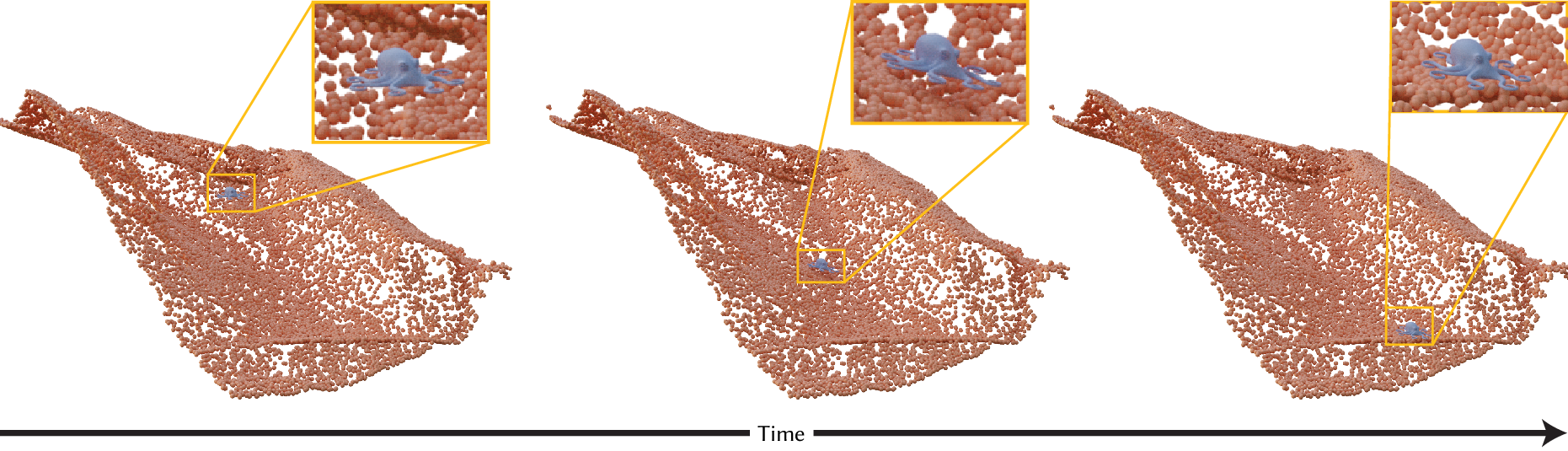}
  \caption{An octopus triangle mesh slides through terrain defined by a lidar point cloud.}
  \label{fig:octopus_lidar}
\end{figure*}

\begin{figure}
  \includegraphics[width=\columnwidth]{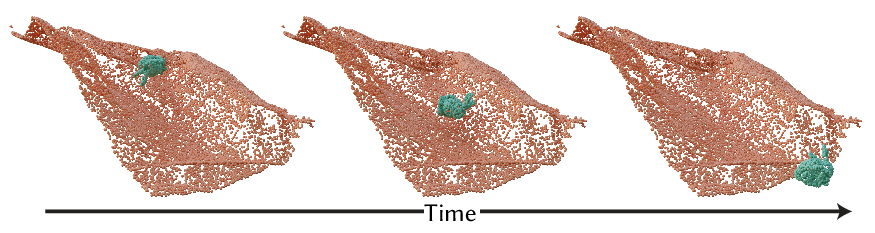}
  \caption{A bunny point cloud rolls and bounces along a lidar point cloud hill.}
  \label{fig:bunny_lidar}
\end{figure}

We summarize the performance results of the 3D simulations in Table~\ref{tbl:sim}, in particular the average time taken to evaluate $\hat{d}(\Md,\Mq)$.
Each evaluation was run with 28 CPU threads parallelizing over per-primitive distance computations.

\begin{table*}
  \centering
  \rowcolors{2}{white}{CornflowerBlue!25}
  \caption{Timing results from the simulations presented in the paper. We give information about the data mesh and query mesh, as well as the average distance evaluation time and the simulation time step size. Each distance evaluation was parallelized over 28 CPU threads.}
  \begin{tabular}{l r r r r r r r r r}
    \toprule
    $\Md$ & \textit{$\Md$ Type} & $|F|$ & $\alpha$ & $\Mq$ & \textit{$\Mq$ Type} & $|\bF|$ & $\alpha_q$ & \textit{Avg. Dist Time (ms)} & \textit{dt (s)} \\
    Bowl & Edges & 736 & 50 & Icosphere & Edges & 120 & 50 & 2.72934 & 1/200 \\
    Bowl & Edges & 736 & 50 & Spiky Sphere & Tris+Edges & 966 & 20 & 34.7789 & 1/1000 \\
    Terrain & Points & 6591087 & 100 & Bunny & Points & 3485 & 20 & 250.0031 & 1/100 \\
    Terrain & Points & 6591087 & 100 & Octopus & Tris & 4432 & 1000 & 189.5989 & 1/100 \\
    Bunny & Tris & 6966 & 200 & Bunny & Tris & 6966 & 200 & 92.5717 & 1/400 \\
    Bumpy Plane & Tris & 800 & 200 & Bunny & Tris & 6966 & 200 & 470.3857 & 1/400 \\
    Slide & Edges & 615 & 100 & Sphere & Tris & 960 & 100 & 33.5601 & 1/100 \\
    Ring Toss & Points & 12400 & 20 & Trefoil & Edges & 100 & 20 & 1.4727 & 1/200 \\
    City Street & Points & 6747648 & 50 & Motorcycle & Tris & 8800 & 100 & 581.2362 & 1/100 \\
    \bottomrule
  \end{tabular}
  \label{tbl:sim}
\end{table*}

\section{Limitations and Future Work}
We have demonstrated a variety of applications of our method, but it is not a panacea.
For example, our method tends to be overly conservative in contact resolution, and may need to be augmented with a more exact method to accurately handle tight contacts like a nut screwing into a bolt.
Furthermore, it can be difficult to select parameters optimally.
Although we can select $\beta$ satisfactorily, $\alpha$ and related parameters require some extra care to select in a way that does not produce \texttt{inf} values and yields highly accurate solutions.
One possibility would be to analyze the distribution of mesh edge lengths to pick a reasonable $\alpha$ that is robust to noise in the distribution.
Furthermore, our weight functions, while effective, are designed heuristically, and it would be interesting to see if there is a more exact way to define them based on the underlying geometry as well as the connectivity.
They are also only applied to the data mesh in the current formulation because our weights require a global minimizer for computing barycentric coordinates of the closest point.
Generalizing weights to work for both data and query meshes, perhaps through a more theoretically grounded weight function, is interesting future work.
Another related direction of future work is analyzing the distribution of a point cloud to eliminate concentration artifacts caused by non-uniform point distributions.
Although this is not a problem for point clouds that come from most lidar scanners, it is a useful property to ensure the generalizability of our method.
Working with noisy point clouds also makes it desirable to eliminate such noise from the dataset entirely.
Smooth distances do not amplify noisy data like exact distances, but the noise still exists in the underlying dataset, which can be problematic particularly at high values of $\alpha$ where the noise begins to separate from the main body into small satellite regions.

Another avenue for future work is applying our method to deforming meshes in both elastodynamics simulations and friction simulation.
Due to $\alpha$'s dependence on the underlying mesh, it would also need to change over time, so it would be interesting to treat $\alpha$ as part of the simulation's evolving state.
Friction is particularly interesting because it involves computing a tangent plane at each contact point, and since contact points are implicitly wrapped in the constraint, such a formulation would require a smooth friction computation over every primitive in the data and query mesh. 
We also believe that there are interesting applications in purely particle-based simulation methods like spherical particle hydrodynamics that have no background grid to store implicit functions, where our method could be dropped in to render the fluid surface or induce collision forces onto the particles from other objects.

Another useful application is to integrate smooth distances into existing simulation frameworks such as IPC~\cite{Li2020} and application extensions~\cite{Fang2021IDP}, replacing multiple exact distance constraints with a single smooth distance constraint.
Although these methods can be implemented using first-order derivative information via gradient descent, their efficient implementations require Hessians to use in Newton iterations.
Smooth distance Hessians inherit discontinuities at projected primitive boundaries from their constituent pairwise primitive distances (though $C^1$ functions are not nearly as problematic as $C^0$ functions --- for example, see the Supplemental of \citet{Li2020}), but the more pressing issue is computing and storing them efficiently, since they are large dense matrices.
It would be worthwhile to investigate how these Hessians can be approximated using techniques like hierarchical matrices.

The explosion of popularity in implicit functions through reconstruction work such as neural SDFs means that improved computational tools for implicit functions are on the horizon, which can also be used to improve the versatility of smooth distances.
For example, computing inertial quantities by directly integrating over the volume enclosed by the zero isosurface would allow simulation frameworks to forget about the underlying geometry almost entirely.

The parallelizability of our method also makes it interesting to consider its integration in purely GPU-based applications.
In particular, there are many possible locations for parallelization; for example, while our current implementation parallelizes the query primitive distance computations, an alternative approach could parallelize the traversals in each per-query primitive evaluation.
We believe analyzing the tradeoffs of these sort of approaches in the style of Halide~\cite{Ragan2013} is promising future work.

\section{Conclusion}

We have presented a smooth distance function that can be efficiently evaluated in such a way that it conservatively estimates the distance to the underlying geometry.
Our function works on various types of geometry that can be encountered: points, line segments, and triangles, and utilizes weight functions to eliminate isosurface artifacts.
We have benchmarked our method on the Thingi10K dataset and shown that it scales quite well for large geometry.
It enables applications such as rigid body contact with lidar point cloud data.
We believe this geometric abstraction is very powerful, and due to its basic preprocessing requirements (simply building a BVH), it can provide a lightweight yet versatile augmentation to the underlying geometry.

\begin{acks}
This research has been funded by in part by NSERC Discovery (RGPIN-2017-05524), Connaught Fund (503114), Ontario Early Researchers Award (ER19-15-034), Gifts from Adobe Research and Autodesk, and the Canada Research Chairs Program.

The authors would like to thank Hsueh-Ti Derek Liu, Silvia Sell\'an, Ty Trusty, Yixin Chen, and Honglin Chen for proofreading, and Silvia Sell\'an and Ty Trusty for help in figure rendering.
The motorcycle, deer, and octopus models are from the Thingi10K dataset; the bunny is from the Stanford 3D Scanning Repository; the city lidar point cloud is from the Toronto-3D dataset; and the terrain point cloud is from OpenTopography.
\end{acks}

\bibliographystyle{ACM-Reference-Format}
\bibliography{smooth-distance}

\appendix

\section{Proof of Underestimate Property}\label{sec:underestimate_proof}

Here we prove that smooth distances underestimate the true distance, when each distance contribution has an associated weight function.

\begin{theorem}
Let $\Md = (V,F)$ be the data mesh, and let $g$ be a query primitive, and each $f_i \in F$ has an associated weight function $w_i(g)$, where $1 \le w_i(g) \le A$ for some constant $A$.
Suppose $d_{min} = \min_{f_i \in F} d(f_i, g)$.
Then, defining $\hat{d}(\Md, g)$ as in Eq.~\ref{eq:lse} (but using a general query primitive $g$ instead of a point), we have $d_{min} \ge \hat{d}(\Md, g) \ge d_{min} - \frac{\log A|F|}{\alpha}$ for all $\alpha > 0$.
\end{theorem}

\begin{proof}
Let $k = \argmin_{f_i \in F} d(f_i, g_j)$.
Using this notation, $d_{min} = d_k$.
Then,
\begin{align*}
  d_k &= -\frac{1}{\alpha} \log \left( \exp \left( -\alpha d_k \right) \right) \\
      &\ge -\frac{1}{\alpha} \log \left( \exp \left( -\alpha d_k \right) \right) - \frac{1}{\alpha} \log w_k(g) \tag{$w_k(g) \ge 1$} \\
      &= -\frac{1}{\alpha} \log \left( w_k(g) \exp \left( -\alpha d_k \right) \right) \\
      &\ge -\frac{1}{\alpha} \log \left( \sum_{f_i \in F} w_i(g) \exp(-\alpha d(f_i, g)) \right) \tag{$\hat{d}(\Md, g)$} \\
      &\ge -\frac{1}{\alpha} \log \left( \sum_{f_i \in F} A \exp(-\alpha d(f_i, g)) \right) \tag{$w_i(g) \le A$} \\
      &\ge  -\frac{1}{\alpha} \log\left( \sum_{f_i \in F} A \exp \left( -\alpha d_k \right) \right) \\
      &= d_k - \frac{\log A|F|}{\alpha}
\end{align*}
Therefore, we have $d_{min} \ge \hat{d}(\Md, g) \ge d_{min} - \frac{\log A|F|}{\alpha}$.
\end{proof}

\section{Exact Distance Formulation}\label{sec:exact_dist}

Here we discuss our distance formulation between simplices $f$ and $g$.
Denoting barycentric coordinates of $f$ as $\bphi$ and $g$ as $\blambda$, the points on each simplex referenced by the barycentric coordinates are denoted as $f(\bphi)$ and $g(\blambda)$ respectively.
(For points, the barycentric coordinate vector is 0-dimensional, so the barycentric coordinate vector can be omitted.)
For example, if $f$ is a triangle consisting of 3 vertices $\bv_0, \bv_1, \bv_2$, and $\bphi = [ \phi^1, \phi^2 ]^\top$, $f(\bphi) = (1-\phi^1 - \phi^2)\bv_0 + \phi^1\bv_1 + \phi^2\bv_2$.
Furthermore, we are restricted to convex combinations of $\bv_i$ (i.e., $0 \le \phi^i \le 1$, and $\sum_i \phi^i \le 1$).
Since any point on a simplex can be referenced using barycentric coordinates, we can determine the closest point pair on $f$ and $g$ by minimizing the squared distance between all points on each simplex:
\begin{equation}\label{eq:distsq}
d^2(f, g) = \min_{\bphi^*, \blambda^*} \lVert f(\bphi^*) - g(\blambda^*) \rVert^2,
\end{equation}
where $\bphi$ and $\blambda$ are the argmin of the problem.
Combined with the aforementioned constraints on $\blambda$ and $\bphi$, we have a quadratic program with linear constraints (and in the case of points, the problem simplifies to the $L^2$ norm).

We are also interested in taking derivatives of this distance.
It is well-known that quadratic programs are differentiable~\cite{hadigheh2007sensitivity}, and in this case, the distance gradient (with respect to $g$) is equivalent to the distance gradient of the closest point pair $f(\bphi)$ and $g(\blambda)$ with respect to $g(\blambda)$, which is
\begin{equation}\label{eq:dgrad}
  \nabla d(f, g) = \begin{cases}
    \frac{g(\blambda) - f(\bphi)}{\lVert g(\blambda) - f(\bphi) \rVert} & d(f, g) \ne 0 \\
    0 & \text{otherwise.}
  \end{cases}
\end{equation}
Although there is a gradient discontinuity when $f$ and $g$ intersect, this is not an issue in our simulation applications, because smooth distance constraints ensure that the underlying geometries will never touch.

This formulation is also tied to the closest point projection function used to define weight functions in Section~\ref{sec:weights}.
Given a closest point projection $\bpi$ onto simplex $f$, $\bpi(g) = \bpi(g(\blambda)) = f(\bphi)$.

\section{Weight Functions}\label{sec:weight_app}
Here we provide more details on the construction and storage of weight functions $w_i(\bq)$ for edges and triangles.

\subsection{Pointwise Case Analysis}
To identify appropriate point constraints for the weight polynomials, we derive the \textit{exact} weights in various cases depending on the location of the closest point on a data mesh $\Md = (V,F)$ to point $\bq$.
First, suppose the closest point is on a vertex $\bv_k \in V$.
Looking at $\hat{d}$, and assuming $\alpha$ is large, we have
\begin{align*}
  \hat{d}(\Md, \bq) &= -\frac{1}{\alpha} \log \left( \sum_{f_i \in F} w_i(\bq) \exp(-\alpha d(f_i, \bq)) \right) \\
                    &\approx -\frac{1}{\alpha} \log \left( \left( \sum_{f_i \in \mathcal{N}_k} w_i(\bq) \right) \exp(-\alpha d(\bv_k, \bq)) \right) \\
                    &= d(\bv_k, \bq) - \underbrace{\frac{\log \left( \sum_{f_i \in \mathcal{N}_k} w_i(\bq) \right)}{\alpha}}_{r_k}
\end{align*}
In the above derivation, $\mathcal{N}_k$ denotes the set of one-ring neighbours of $\bv_k$, and the second line is a consequence of high $\alpha$ making the other terms in the summation negligible.
Our goal is to ensure that $r_k = 0$, which is equivalent to the condition $\sum_{f_i \in \mathcal{N}_k} w_i(\bq) = 1$.
We can accomplish this by setting $w_i(\bq) = \frac{1}{|\mathcal{N}_k|}$ if $\bq$'s closest point to $f_i$ is its vertex $\bv_k$ (which will be the case for all $f_i \in \mathcal{N}_k$ when $\bq$'s closest point on $\Md$ is $\bv_k$).
When $\Md$ is a triangle mesh and the closest point is along an edge $(k, \ell)$, we can use an identical argument to derive $w_i(\bq) = \frac{1}{|\mathcal{N}_{k\ell}|}$ where $\mathcal{N}_{k\ell}$ denotes the set of triangles incident on $(k,\ell)$.

Now suppose the closest point is on a simplex $f_i \in F$ (but not on one of its faces, in which case the earlier discussion applies).
Again assuming $\alpha$ is large, we have
\begin{align*}
  \hat{d}(\Md, \bq) &= -\frac{1}{\alpha} \log \left( \sum_{f_i \in F} w_i(\bq) \exp(-\alpha d(f_i, \bq)) \right) \\
                    &\approx -\frac{1}{\alpha} \log \left( w_i(\bq) \exp(-\alpha d(f_i, \bq)) \right) \\
                    &= d(f_i, \bq) - \frac{\log w_i(\bq)}{\alpha}
\end{align*}
In this case, we can simply set $\tilde{w}_i(\bq) = 1$.

Of course, we cannot directly use these exact weights, since they are not smooth and do not satisfy the conservative property $w_i(\bq) \ge 1$.
Instead, we use these cases as guidelines for point constraints in building non-conservative polynomial weights $\tilde{w}_i$ and computing an appropriate scale factor $A$ to ensure the conservative property holds.

To simplify the notation in the remainder of this section, we will work with barycentric coordinates of the closest point projection of $\bq$ onto $f_i$ (i.e., $\tilde{w}_i(\bphi_i)$).

\subsection{Edge Weights}
Edges have a single barycentric coordinate, so the weight function of an edge $f_i$ in terms of barycentric coordinates is $\tilde{w}_i(\phi_i)$.
The weight function is a 4th order polynomial, with point constraints $\tilde{w}_i(0) = \frac{1}{|\mathcal{N}_{i_0}|}$, $\tilde{w}_i(1) = \frac{1}{|\mathcal{N}_{i_1}|}$, and $\tilde{w}_i(0.5) = 1$, and derivative constraints $\tilde{w}_i(0) = \tilde{w}_i(1) = 0$.
Using the polynomial coefficients as unknowns, the five above equations become a linear system that can be solved for unique coefficients that satisfy the constraints.
We have set the three extrema of this quartic function by construction, which allows us to cheaply compute $A = \max_k |\mathcal{N}_k|$.

\begin{figure*}[t!]
  \includegraphics[width=\textwidth]{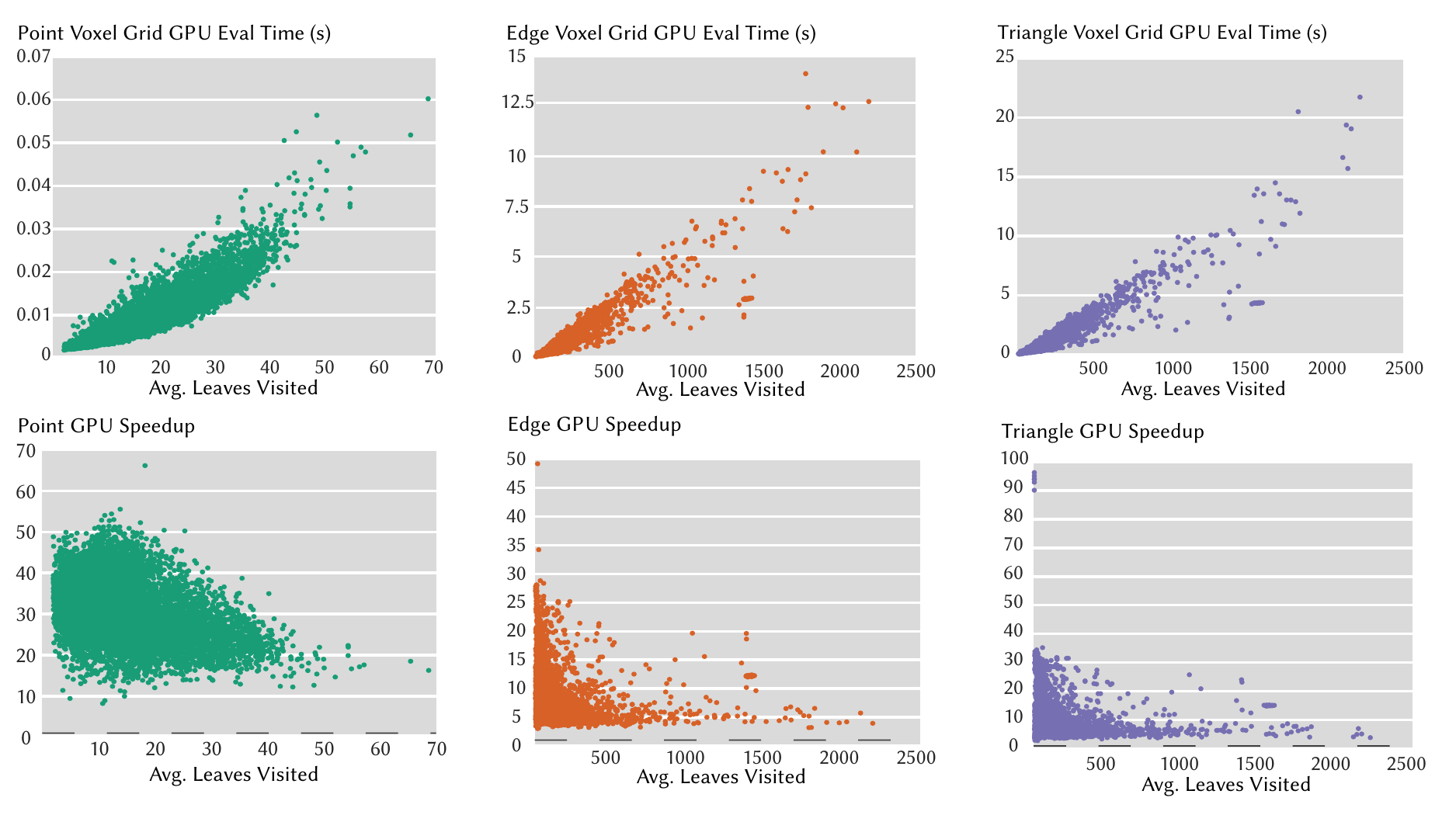}
  \caption{Results of the Thingi10K benchmark run on the GPU, with the same test as in Fig.~\ref{fig:benchmark}. The top row shows times and average leaves visited, while the bottom row shows speedups relative to the CPU benchmark. The dotted lines represent a speedup of 1 (i.e., identical performance on the CPU and GPU).}
  \label{fig:gpu_benchmark}
\end{figure*}

\subsection{Triangle Weights}
Triangles have two barycentric coordinates, so the weight function of a triangle $f_i$ in terms of barycentric coordinates is $\tilde{w}_i(\phi_i^1, \phi_i^2)$.
In order to satisfy the constraints, $\tilde{w}_i$ is a 7th order polynomial, with 36 coefficients.
Since we must store a set of coefficients for each triangle in a mesh, we additionally require that the function is symmetric, $\tilde{w}_i(\phi_i^1, \phi_i^2) = \tilde{w}_i(\phi_i^2, \phi_i^1)$, so that we can roughly halve the number of coefficients that must be stored.

We use 10 point constraints: one per vertex, two per edge (spaced equally along their lengths), and one for the triangle barycenter.
In order to ensure that $\tilde{w}_i$ is symmetric, we assign the same value at each vertex constraint and at each edge constraint, based on the maximum vertex and edge valence, respectively.
More concretely, we first find $v_i = \max \{ |\mathcal{N}_{i_0}|, |\mathcal{N}_{i_1}|, |\mathcal{N}_{i_2}| \}$ and $e_i = \max \{ |\mathcal{N}_{i_{01}}|, |\mathcal{N}_{i_{12}}|, |\mathcal{N}_{i_{20}}| \}$ (noting that $v_i \ge e_i$),
and assign $\tilde{w}_i(0,0) = \tilde{w}_i(1,0) = \tilde{w}_i(0,1) = \frac{1}{v_i}$, and $\tilde{w}_i(1/3,0) = \tilde{w}_i(2/3,0) = \tilde{w}_i(1/3,2/3) = \tilde{w}_i(2/3,1/3) =\tilde{w}_i(0,1/3) = \tilde{w}_i(0,2/3) = \frac{1}{e_i}$.
Without carefully setting the final barycenter constraint, it is possible to produce spurious local extrema in regions other than the vertices, edges, and barycenter, which is problematic for determining an appropriate $A$, and can even lead to phenomena like negative weights.
We have found that a barycenter constraint of $\tilde{w}_i(1/3,1/3) = \frac{v_i - 1}{v_i}$ prevents this spurious extrema issue; although the barycenter weight is no longer 1, the weight at the barycenter is still significantly higher than at the vertices, so the deviation from exact weights is acceptable.
Then, like with edges, we have $A = \max_i v_i$.

The normal derivative constraints, are more complex than for edges: we must enforce constraints on the lines $\phi_i^1 = 0$, $\phi_i^2 = 0$, and $\phi_i^1 + \phi_i^2 = 1$ within the barycentric triangle.
Taking care to ensure symmetry, these constraints are $\frac{\partial}{\partial \phi_i^1} \tilde{w}_i(0,t) = \frac{\partial}{\partial \phi_i^2} \tilde{w}_i(t,0) = 0$ for all $t \in [0,1]$, and
$\left( \frac{\partial}{\partial \phi_i^1} + \frac{\partial}{\partial \phi_i^2} \right) \tilde{w}_i(t,1-t) = \left( \frac{\partial}{\partial \phi_i^1} + \frac{\partial}{\partial \phi_i^2} \right) \tilde{w}_i(1-t,t) = 0$ for all $t \in [0,1]$.
Essentially, each constraint equation is assigning a 6th order univariate polynomial to be zero over a line segment, which can only happen with the zero polynomial.
Thus, each coefficient of these polynomials, which is a linear combination of coefficients of $\tilde{w}_i$, must be 0 as well, creating 7 new linear equations per constraint.
Noting that the 1st order coefficients of $\frac{\partial}{\partial \phi_i^1} \tilde{w}_i(0,t)$ and $\frac{\partial}{\partial \phi_i^2} \tilde{w}_i(t,0)$ are equal,
as well as the 6th order coefficients of $\left( \frac{\partial}{\partial \phi_i^1} + \frac{\partial}{\partial \phi_i^2} \right) \tilde{w}_i(t,1-t)$ and $\left( \frac{\partial}{\partial \phi_i^1} + \frac{\partial}{\partial \phi_i^2} \right) \tilde{w}_i(1-t,t)$, we have 26 unique equations, leading to a total of 36 equations.
This system of equations still has a non-trivial null space, but we can solve for reasonable weight coefficients using the pseudoinverse.

To further reduce the memory footprint of these weights, we observe that the constraint equations corresponding to $\phi_i^1 = 0$ and $\phi_i^2 = 0$ are each in terms of a single $\tilde{w}_i$ coefficient, so we do not need to store those 13 coefficients at all.
Combined with symmetry, we only need to store 13 unique coefficients per weight function.

\section{Far Field Approximation}\label{sec:farfield}

Here we describe the Barnes-Hut far field approximation in more detail.
Let $B$ be a bounding region of points, $n_B$ be the number of points in $B$, $|B|$ be the diameter of $B$ (e.g., the bounding box diagonal), and $\by_B$ be the center of the far field approximation of $B$.
If $\frac{|B|}{d(\by_B, g)} < \beta$ for some user-defined $\beta$, then approximate the exponential summation over points $B$ as $\sum_{f_i \in B} \exp(-\alpha d(f_i, g)) \approx n_B \exp(-\alpha d(\by_B, g))$.
Similarly, the gradient contribution of $B$ (in the numerator summation of Eq.~\ref{eq:dlse}) is now $n_B \exp(-\alpha d(\by_B, g)) \frac{g(\blambda) - \by_B}{\lVert g(\blambda) - \by_B \rVert}$.
When our data primitives are edges or triangles, we must take some care with the weight function, which must be incorporated to ensure the approximation is reasonably smooth.
Since the approximation is just a single point, there is no undesirable concentration in some regions and we can simply use a weight of $w = A^S$, which corresponds to a constant weight function scaled by $A$ and attenuated by $S$.

Viewing the far field approximation as a Taylor series, we have only included the constant term.
However, we found that higher-order Taylor series terms tended to produce worse results due to the off-center expansion point amplifying the error of those terms.

\section{GPU Benchmark}\label{sec:gpu_benchmark}

Here we present the results of our GPU benchmark.
Structurally, it is identical to the benchmark in Section~\ref{sec:benchmark}, but it is run on a GPU instead of several CPU threads.
We ran these tests on a Titan RTX, and implemented our method in CUDA by ensuring our CPU code was also GPU-compatible.
The results are presented in Fig.~\ref{fig:gpu_benchmark}.
We see that, while performance is still linearly proportional to leaves visited, and the GPU code consistently outperforms CPU code by staying over the speedup=1 line and achieving an average $30\times$ speedup on points and $5\times$ speedup on edges and triangles, the relative speedups vary greatly, and are significantly reduced as more leaves are visited.
This is likely because global memory accesses become a significant issue in these cases, on top of the additional computation.
This is an especially pronounced issue on GPUs since memory is much slower to access on a GPU than a CPU.

\end{document}